\documentclass[journal]{IEEEtran}
\usepackage{amsmath,amsfonts}
\usepackage{algorithmic}
\usepackage{algorithm}
\usepackage{array}
\usepackage[caption=false,font=normalsize,labelfont=sf,textfont=sf]{subfig}
\usepackage{textcomp}
\usepackage{stfloats}
\usepackage{url}
\usepackage{verbatim}
\usepackage{graphicx}
\usepackage{cite}

\usepackage{hyperref}
\hypersetup{
  colorlinks=true, 
  linkcolor=blue,   
  citecolor=blue,  
  urlcolor=blue     
}

\usepackage{multirow}
\usepackage{adjustbox}
\usepackage{upgreek} 
\usepackage{diagbox} 

\usepackage{xcolor}
\definecolor{SEGreen}{RGB}{20, 200, 60}
\usepackage{pifont} 

\newcommand{\R}{\mathbb{R}} 
\usepackage{amsthm} 
\usepackage{amssymb} 

\usepackage{stmaryrd} 
\newtheorem{theorem}{Theorem}

\DeclareMathOperator{\ReLU}{ReLU}
\DeclareMathOperator{\NR}{NR}
\DeclareMathOperator{\erf}{erf}
\usepackage[sort&compress,numbers]{natbib} 
\usepackage{booktabs} 
\renewcommand{\arraystretch}{1.2} 

\begin{document}

\title{Topology-Aware Graph Neural Network-based State Estimation for PMU-Unobservable Power Systems}

\author{Shiva Moshtagh,~\IEEEmembership{Student Member,~IEEE}, Behrouz Azimian,~\IEEEmembership{Student Member,~IEEE}, Mohammad Golgol,~\IEEEmembership{Student Member,~IEEE}, and Anamitra Pal,~\IEEEmembership{Senior Member,~IEEE}
\thanks{This work was supported in part by the Department of Energy under grant DE-EE0009355 and National Science Foundation under grant ECCS-2145063.

The authors are with the School of Electrical, Computer, and Energy Engineering (ECEE) of Arizona State University (ASU). Emails: \url{smoshta1@asu.edu}; \url{bazimian@asu.edu}; \url{mgolgol@asu.edu}; \url{Anamitra.Pal@asu.edu}
}
}

\markboth{IEEE Transactions on Power Systems}%
{Shell \MakeLowercase{\textit{et al.}}: A Sample Article Using IEEEtran.cls for IEEE Journals}


\maketitle

\begin{abstract}
Traditional optimization-based techniques for \textit{time-synchronized} state estimation (SE) often suffer from high online computational burden, limited phasor measurement unit (PMU) coverage, and presence of non-Gaussian measurement noise.
Although conventional learning-based models have been developed to overcome these challenges, they are negatively impacted by topology changes and real-time data loss. 
This paper proposes a novel deep geometric learning approach based on graph neural networks (GNNs) to estimate the states of PMU-unobservable power systems. 
The proposed approach combines graph convolution and multi-head graph attention layers inside a customized end-to-end learning framework to handle topology changes and real-time data loss.
An upper bound on SE error as a function of topology change is also derived.
Experimental results for different test systems demonstrate superiority of the proposed customized GNN-SE (CGNN-SE) over traditional optimization-based techniques as well as conventional learning-based models in presence of topology changes, PMU failures, bad data, non-Gaussian measurement noise, and large system implementation.
\end{abstract}

\vspace{-1em}


\section{Introduction}
\label{Introd}
\IEEEPARstart{S}{tatic} state estimation (SE) plays a pivotal role in energy management systems as it provides inputs to critical applications such as real-time contingency analysis and dynamic security assessment \cite{biswas2021mitigation,hu2023high}.
Considering the need for high-speed (sub-second) situational awareness in modern inverter-based resource-rich power systems \cite{chatzivasileiadis2023micro}, there is growing interest in phasor measurement unit (PMU)-assisted hybrid and PMU-only linear state estimators \cite{dvzafic2020real,mishra2022algebraic}.
These estimators, collectively termed \textit{optimization-based techniques} for SE in this paper,  have been the focus of numerous studies (e.g., see \cite{cheng2023survey}).
However, PMU-assisted hybrid state estimators can be 
affected by imperfect synchronization and time-skew errors \cite{zhao2018robust,darmis2023survey}.
Furthermore, the strategies proposed to mitigate these issues \cite{yang2013power,zhao2015power,manousakis2018hybrid} are generally computationally demanding, which causes hybrid state estimators to function at slower timescales \cite{jin2018analysis}.
In contrast, PMU-only linear SE (LSE) requires complete system observability by PMUs \cite{chen2022optimal}.
As noted in the first two tables of \cite{varghese2024deep}, many power utilities still lack complete PMU coverage of their bulk power systems.
Moreover, optimization-based techniques for SE often have high online computational burden and/or limited robustness against non-Gaussian measurement noise \cite{tian2021neural,azimian2021time}.

Considering these limitations of 
optimization-based SE techniques, there has been a significant emphasis on
developing \textit{learning-based approaches} for SE.
The most popular learning-based approach is the conventional deep neural network (DNN)-based SE (DNN-SE). The ability of conventional DNNs to perform 
\textit{time-synchronized} 
SE in PMU-unobservable power systems has been demonstrated in \cite{varghese2024deep,tian2021neural,azimian2021time,mestav2019learning}.
By relegating training to the offline stage, conventional time-synchronized DNN-SE
is able to provide sub-second situational awareness. Similarly, the excellent approximation and generalization 
capabilities of DNNs enable DNN-SE to handle non-Gaussian noise effectively.
However, a DNN-SE is not able to account for
structural changes occurring in the power system (e.g., line opening)
and/or the sensing system (e.g., PMU failure). These two aspects are elaborated below. 

Transmission lines can be opened in response to forced outages or maintenance \cite{huang2015dynamic} as well as for optimizing dispatch \cite{li2017real}.
As an example, Fig. \ref{TVA_branches} shows how the total number of branches in the transmission system of a power utility located in the U.S. Eastern Interconnection
changed over a period of six months.
The figure corresponds to half-hourly state estimator snapshots taken during the same hour of every day.
Notably, the figure reveals a significant variation in the number of branches between two consecutive snapshots (an average change of $\approx2$ branches).
Therefore, if a DNN-SE is not able to account for topology changes, then its performance will degrade considerably over time.

One way to address this issue is to create a DNN for detecting topology changes at high-speeds (such as the one developed in \cite{gotti2021deep}), and use it as a precursor to a DNN-SE.
Then, the DNN-SE can be re-trained via transfer learning after a topology change occurs, as was done for distribution systems in \cite{azimian2022state}.
However, the framework developed in \cite{azimian2022state} could only detect changes in switch statuses and not line openings. 
Moreover, the re-training took a few minutes, which 
may be insufficient for transmission system applications \cite{chatzivasileiadis2023micro}.
Finally, having two DNNs in series could potentially degrade overall performance as the errors of the first DNN will impinge on the next.
On a similar note, a structural change in the sensing system, resulting from a PMU failure during online operation, may necessitate training of the DNN-SE from scratch due to a mismatch in the inputs.
Considering these aspects of the problem, it is essential to explore machine learning (ML)-based SE strategies that can handle line openings and PMU failures implicitly, i.e., without additional pre-/post-processing.

\begin{figure}[ht]
\centering
\includegraphics[width=0.48\textwidth]{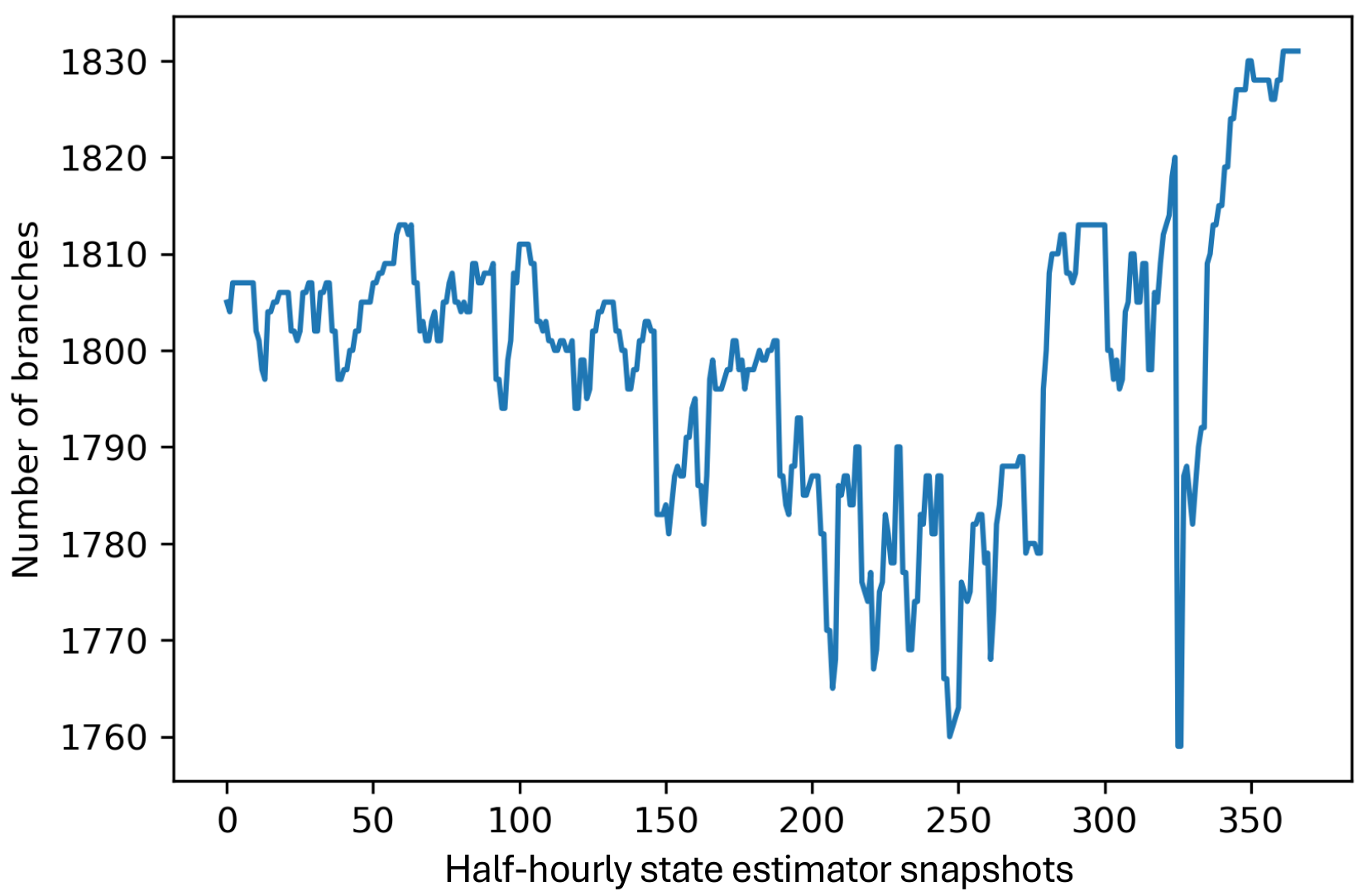}
\vspace{-0.5em}
\caption{Total number of branches in a utility power system located in the U.S. Eastern Interconnection across 360 SE snapshots (= two-per-hour-per-day for six months).}
\vspace{-1em}
\label{TVA_branches}
\end{figure}

To overcome the drawbacks of conventional DNNs, \textit{structural learning} has been introduced by the ML community, with the most relevant structural learning architecture for power systems being the graph neural network (GNN) \cite{liao2021review}.
GNNs excel at capturing intricate relationships and dependencies among nodes in a graph by aggregating data from neighboring nodes. 
The inputs of a GNN include feature-based and structure-based information of the graph appearing in the form of the node feature matrix and the adjacency matrix, respectively. Hence, the learning process of the GNN, known as \textit{geometric learning}, incorporates both the data and the physical information of the graph, making it distinct from conventional DNNs that only rely on data \cite{bronstein2017geometric}.

GNNs have been employed for a variety of power system applications.
In the following, we summarize their use in the \textit{transmission system static SE application}, which is the focus of this paper.
A GNN-based regularizer was used as a prior in \cite{yang2022data} to create a SE formulation that was robust against bad data and adversarial data.
Conversely, a DC power flow was used as a prior in \cite{wu2022integrating} to create the node feature and edge feature matrices of the GNN, which was then used to perform SE. However, the state estimators developed in \cite{yang2022data,wu2022integrating} did not use PMU data, implying that they cannot provide sub-second situational awareness.
Refs. \cite{de2022physics,kundacina2023graph,10318579} used PMU data as inputs for GNN-based SE (GNN-SE).
However, these formulations required the system to be completely observable by PMUs 
as that ensured satisfaction of the \textit{fully-observed node feature matrix} condition of GNNs  \cite{chen2022learning}.
To bypass this requirement,
use of feature propagation and augmentation techniques have been explored \cite{rossi2022unreasonable,liu2022local}.
However, such processes for finding missing features are separate from geometric learning, which leads to unstable and less reliable performance \cite{taguchi2021graph}.

In this paper, we integrate distributions learnt from historical supervisory control and data acquisition (SCADA) data into an \textit{end-to-end geometric learning process}\footnote{In this process, missing feature construction is integrated directly into geometric learning without separate pre-processing.}
to address the PMU-unobservability problem.
Subsequently, we verify its robustness to topology changes by deriving an upper bound on its maximum error as a function of the change in topology.
The salient features of the proposed \textit{customized} GNN-SE (CGNN-SE)
framework are as follows:

\begin{itemize}
\item \textit{High-speed:} The proposed framework is fast because during online operation, it solely utilizes PMU data and does not require iterative calculations/matrix inversions.
\item \textit{Overcomes unobservability:} The proposed framework addresses unobservability issues caused by limited PMU coverage and/or real-time PMU failure by incorporating distributions learnt from historical SCADA data into an end-to-end deep geometric learning process.
\item \textit{Robust to topology changes:}
Different from an approach where the ML-based state estimator must be retrained for the new topology, 
the proposed framework enables accurate and consistent SE even after a topology change takes place, without the need for subsequent retraining.
\item \textit{Provision of explainability:} The topology adaptivity of the proposed framework is explained mathematically.
\item \textit{Handles non-Gaussian noise:} The deep geometric learning process of the proposed framework enables it to better handle non-Gaussian measurement noise in comparison to traditional optimization-based state estimators.
\item \textit{Handles bad PMU data:} The 
proposed framework is incorporated with a bad data detection and correction scheme that mitigates the negative impact of bad data.
\item \textit{Scalable:} The scalability of the proposed framework is demonstrated by applying it
to large-scale systems.
\end{itemize}

The rest of the paper is structured as follows. Section \ref{GNN} presents an overview of GNN in SE context, Section \ref{unobservability} introduces the proposed CGNN-SE framework to overcome the unobservability issue, and provides an analytical explanation to show the topology adaptivity of CGNN-SE. Section \ref{results} showcases numerical results of the performance of the proposed CGGN-SE compared to traditional optimization-based techniques as well as conventional learning-based models on different test power systems, and Section \ref{Conclusions} highlights concluding remarks and future work.


\section{An overview of GNN-based SE (GNN-SE)}
\label{GNN}
The electric power grid can be represented by a graph,
$\mathcal{G}=(\mathcal{V},\mathcal{E})$, where $\mathcal{V}=\{v_1,\hdots, v_N\}$ denotes the set of nodes (vertices/buses), and $\mathcal{E}=\{(v_i,v_j)\}\subseteq\mathcal{V}\times\mathcal{V}$ represents the set of edges (transmission lines and transformers).
Each node $v_i$ in the graph is associated with a feature vector, $x_i \in \R^{f} $, where $f$ signifies the number of features for each node. These feature vectors are stacked to create the node feature matrix, $X \in \R^{N \times f}$, that represents the feature-based information of the whole graph.
In the context of time-synchronized GNN-SE, $X$ is created from PMU measurements\footnote{When PMU coverage is limited or there is sudden PMU failure, there will be missing features in the node feature matrix. This is a major limitation of the GNN-SE formulations developed in \cite{de2022physics,kundacina2023graph,10318579}.}.
Regarding the structural information of the graph, the adjacency matrix,  $A \in \R^{N \times N}$, is employed which represents the
node connections within the grid, i.e., $a_{ij} = 0$ if  $(v_i, v_j) \notin \mathcal{E}$ and $a_{ij} = 1$ if $(v_i, v_j) \in \mathcal{E}$.
In the following, we explain how the adjacency matrix and a fully-observed node feature matrix contribute to the learning process of a conventional GNN-SE.

A GNN learns via a layered architecture. At each layer, node features are updated by aggregating information from neighboring nodes using \textit{message-passing}. This involves computing a weighted sum of neighboring node features, which are then combined with the node's own features to produce an updated vector. This process captures local connections in the grid, embedding them into the learning process. Updated feature vectors are now passed to the next GNN layer. Finally, the GNN's output (final node representation) is used for prediction purposes, e.g., SE. Commonly used aggregation functions in GNNs are briefly summarized below.


\subsection{Graph Convolutional Network (GCN)}
\label{GCN}
A GCN is a simple and computationally efficient aggregation function that captures the
local graph structure through a \textit{mean} aggregation scheme as shown below:
\begin{equation}\label{eq:GCN1}
x^{(l)}_i=\sigma \left( \sum_{j\in\mathcal{N}(i)\cup\{i\} }W^{(l)}\frac{x^{(l-1)}_j}{\sqrt{|\mathcal{N}(i)||\mathcal{N}(j)|}} \right ) 
\end{equation}
where, $x^{(l)}_i \in \mathbb{R}^{f^{(l)}}$ and $x^{(l-1)}_j \in \mathbb{R}^{f^{(l-1)}}$ are feature vectors of nodes $v_i$ and $v_j$ at layers $l$ and $l-1$ with $f^{(l)}$ and $f^{(l-1)}$ number of features, respectively, $\sigma(\cdot)$ represents the rectified linear-unit (ReLU) activation function, $W^{(l)} \in \mathbb{R}^{f^{(l)} \times f^{(l-1)}}$ is the learnable weight matrix at the $l^{th}$ layer, $|\mathcal{N}(i)|$ and $|\mathcal{N}(j)|$ are the degrees of nodes $v_i$ and $v_j$, respectively, and $l=1,2,\hdots, \mathrm{L}$ refers to different layers of the GNN. 
It should be noted that the size of $W^{(l)}$ is a hyperparameter that can be adjusted at each layer $l$ to control the expressiveness and complexity of the learned node representations.

Equation \eqref{eq:GCN1} can be written in matrix form to represent the whole graph as shown below:
\begin{equation}\label{eq:GCN2}
X^{(l)}=\sigma\left(\Tilde{A}X^{(l-1)}{W^\top}^{(l)}\right)
\end{equation}
where, $\Tilde{A}=D^{-1/2}\hat{A}D^{-1/2}$ is the graph Laplacian, $\hat{A}=A+I$ represents the adjacency matrix with inserted self-loops to account for each node’s own feature vector in the update process, and $D=\text{diag}\left(\sum_i\hat{a}_{1i}, \dots, \sum_i\hat{a}_{Ni}\right)$ is the degree matrix representing the number of neighbors of each node, where $\hat{a}_{ij} \in \hat{A}$.


\subsection{Graph Attention Network (GAT)}\label{GAT}
A GAT improves upon GCN by introducing attention mechanism which offers the advantage of
\textit{adaptive} neighborhood aggregation 
\cite{velickovic2017graph}
Unlike GCN which treats all neighbors as equally important, a GAT layer adaptively and implicitly adjusts the contribution of each neighbor. 
It does this by assigning different attention weights to each neighbor node when aggregating information, and computes attention coefficients for each pair of nodes using the following equation:

\begin{equation}\label{eq:GAT1}
e^{(l)}_{ij}=\sigma_{\mathrm{leaky}}\left(\mathrm{a}^\top\left[W^{(l)}x^{(l-1)}_i \bigparallel W^{(l)}x^{(l-1)}_j\right]\right)
\end{equation}

In \eqref{eq:GAT1}, $e^{(l)}_{ij}$ is the attention coefficient between nodes $v_i$ and $v_j$ at layer $l$,  $\mathrm{a} \in \R^{2f^{(l)}}$ is a learnable parameter vector, $W^{(l)} \in \mathbb{R}^{f^{(l)} \times f^{(l-1)}}$ represents a learnable weight matrix at layer $l$, $\bigparallel$ denotes the concatenation operation, and ${\sigma_{\text{leaky}}}(\cdot)$ is the leaky-ReLU activation function.
Attention coefficients are then normalized using the $\mathrm{softmax}$ function shown below:
\begin{equation}\label{eq:GAT2}
\alpha^{(l)}_{ij}=\mathrm{softmax}\left(e^{(l)}_{ij}\right) = \frac{\mathrm{exp}\left(e^{(l)}_{ij}\right)}{\sum_{k\in\mathcal{N}(i)\cup\{i\}}\mathrm{exp}\left(e^{(l)}_{ik}\right)}
\end{equation}
where, $\alpha^{(l)}_{ij}$ is the normalized attention weight between nodes $v_i$ and $v_j$ at layer $l$.
The final aggregation function by a GAT layer is written as:
\begin{equation}\label{eq:GAT3}
x^{(l)}_i=\sigma \left(\sum_{j\in\mathcal{N}(i)\cup\{i\}}\alpha^{(l)}_{ij}W^{(l)}x^{(l-1)}_j \right)
\end{equation}
where, ${x^{(l)}_i}$ represents the updated representation of node $v_i$ in the $l^{th}$ GAT layer. It should be noted that the learnable weight matrix, $W$, governs the linear transformation of feature vectors, allowing the model to adjust how it combines information from the neighbors. 
At the same time,
the learnable attention weights, $\alpha_{ij}$, determined through attention mechanism, control the importance of neighboring nodes. This empowers the GAT layer to selectively emphasize or de-emphasize information from different neighbors.
Equation \eqref{eq:GAT3} can be represented in matrix form similar to \eqref{eq:GCN2}, with the distinction that $\Tilde{A}$ is replaced by the attention weight matrix, where its elements are $\alpha_{ij}$ if $(v_i,v_j) \in \mathcal{E}$.


\subsection{GNN-based SE (GNN-SE) Architecture}\label{GNN-SE}
The GNN-SE framework incorporates a stacked architecture, combining GCN followed by
GAT layers to leverage the strengths of both variants.
After $\mathrm{L}$ GCN and GAT layers, the last layer applies a linear transformation to the final node representations of the $\mathrm{L^{th}}$ layer, $X^{(\mathrm{L})}$, to predict the state variables.
This is mathematically described by: 
\begin{equation}\label{eq:linear}
\hat{Y}=X^{(\mathrm{L})}W^{'}+b^{'}
\end{equation}

In \eqref{eq:linear}, $\hat{Y} \in \R^{N \times s}$ represents the prediction matrix for all $N$ nodes, each of which possesses $s$ state variables, and $W^{'}$ and $b^{'}$ are the trainable parameters of the linear layer.
In our case, $s=2$ since we are predicting voltage magnitude and phase angle of every node in the power system graph. 

Since the GNN has the adjacency matrix as one of its inputs, it becomes aware of the physical connections of the power system during its learning process. 
This enables it to better handle topology changes \cite{10318579}.
However, as mentioned previously in Section \ref{Introd}, for performing GNN-SE, a fully-observed node feature matrix is required.
Since real-world power systems are often incompletely observed by PMUs \cite{varghese2024deep}, we introduce a novel GNN-SE framework in the next section that relaxes this requirement.


\section{Proposed Customized GNN-SE (CGNN-SE)}\label{unobservability}
We start by customizing the first layer of the GNN-SE described in Section \ref{GNN} (i.e., a GCN) by integrating historical SCADA data in the form of Gaussian mixture models (GMMs) into the geometric learning process of the GNN.
After representing the missing features with GMMs, we calculate the expected activation of nodes in the first hidden layer of CGNN-SE while keeping the remaining layers the same as GNN-SE.
This enables it to learn the GMM parameters and network weight parameters in an end-to-end manner without significantly increasing the computational complexity \cite{taguchi2021graph}.
The initial values of the GMM parameters are found using expectation-maximization (EM).
Next, we replace the GAT layer with a multi-head GAT (MH-GAT) layer to better capture the diverse nodal relationships using the concept of multiple attention mechanisms \cite{huang2021topology}.
Finally, we provide performance guarantees of the CGNN-SE as a function of topology change as well as the overall implementation set-up.

\subsection{Representing Missing Features with GMMs}
To deal with missing features in the power system graph as a result of PMU scarcity (due to lack of PMU coverage and/or real-time PMU data loss), we treat the feature vector $x_i \in {\mathbb{R}}^{f}$ introduced in Section \ref{GNN} as a random variable.
This random variable can be represented using a mixture of Gaussians irrespective of whether $x_i$ is missing or not, as shown below:
\begin{equation}\label{eq:GMM1}
x_i \sim \sum_{c=1}^{C} \pi^c_i \mathcal{N}(\mu^c_i,\Sigma^c_i)
\end{equation}
where, $c=1,2,\hdots, \mathrm{C}$ corresponds to the number of Gaussian components, and $\pi^c_i$, $\mu^c_i = \left[\mu_{i1}^{c}, \cdots, \mu_{if}^c\right]^\top$, and $\Sigma^{c}_i = \text{diag}(({\sigma_{i1}^{c}})^2, \cdots, ({\sigma_{if}^{c}})^2)$ are weight, mean, and covariance parameters of the $c^{th}$ Gaussian component for node $v_i$, respectively.
Following this, we can define the node feature matrix containing both missing and known features as:
\begin{equation}\label{eq:GMM2}
X \sim \sum_{c=1}^{C} p^c \mathcal{N}(M^c,S^c)
\end{equation}
where, $p^c = \left[\pi^c_1, \dots, \pi^c_n\right]^\top$ is the weight vector for each component $c$ so that $\sum_{c=1}^{C} \pi^c_i = 1$ holds true for all feature vectors $x_i \in X$. 
Furthermore, $M^c \in \R^{N \times f}$ and $S^c \in \R^{N \times f}$ are the mean and covariance matrices for each component $c$, respectively, whose rows are defined as:
\begin{align}\label{eq:GMM3}
    {m}_{i}^{c} &= \begin{cases}
        x_{i} & \text{if } v_{i} \text{ is equipped with PMU;} \\
        \mu_i^{c} & \text{otherwise}
    \end{cases} \\
    {s}_{i}^{c} &= \begin{cases}\label{eq:GMM4}
        0 & \text{if } v_{i} \text{ is equipped with PMU;} \\
        \Sigma^{c}_i & \text{otherwise}
    \end{cases}
\end{align}

Considering \eqref{eq:GMM3} and \eqref{eq:GMM4}, all node features regardless of being missing or known can be represented as GMM parameters, with the difference being that nodes equipped with PMUs directly provide the mean with zero variance, while nodes with missing features provide the corresponding parameters through fitted GMMs and given historical data. These parameters are then updated during the end-to-end learning process. 
As the first layer of the proposed CGNN-SE is a GCN, we quantify the expected activation of each node $v_i$ by calculating \eqref{eq:GCN2} while the node feature matrix $X$ follows a mixture of Gaussians. 
To calculate the expected activation, Theorem \ref{theorem1} is defined for a general case and Theorem \ref{theorem2} is defined for our specific case.

\begin{theorem}\label{theorem1}
Assuming a Gaussian distribution denoted as $\mathcal{N}\left(\mu, \sigma^2\right)$, we have:
\begin{equation}\label{eq:GMM5}
\ReLU\left[\mathcal{N}\left(\mu, \sigma^2\right)\right]=\sigma \NR\left(\frac{\mu}{\sigma}\right)
\end{equation}
where, $\NR(\cdot)$ is an auxiliary function defined as:
\begin{equation}\label{eq:GMM6}
\NR(z)=\frac{1}{\sqrt{2\pi}}\exp\left(-\frac{z^2}{2}\right)+\frac{z}{2}\left(1+\erf \left( \frac{z}{\sqrt{2}} \right) \right)
\end{equation}
and $\erf(\cdot)$ is the error function as:
\begin{equation}\label{eq:GMM7}
\erf(z)=\frac{2}{\sqrt{\pi}} \int_{0}^{z} \exp\left(-t^2\right)dt.
\end{equation}
\end{theorem}
\begin{proof}
Please see \cite{smieja2018processing} for the proof.
\end{proof}

\begin{theorem}\label{theorem2}
If $X \sim \sum_{c=1}^{C} p^c \mathcal{N}(M^c,S^c)$, then given the matrices $A$ and $W$:
\begin{equation}\label{eq:GMM8}
\ReLU\left[(AXW)_{ij}\right]=\sum_{c=1}^{C}\pi^c_i\sqrt{\hat{s}_{ij}^{c}}\NR \left(\frac{\hat{m}_{ij}^{c}}{\sqrt{\hat{s}_{ij}^{c}}} \right)
\end{equation}

where, $\hat{m}_{ij}^{c}$ and $\hat{s}_{ij}^{c}$ are elements of $\hat{M}^{c}$ and $\hat{S}^{c}$, respectively, such that:
\begin{equation}\label{eq:GMM9}
\hat{M}^{c} = AM^{c}W
\end{equation}
\begin{equation}\label{eq:GMM10}
\hat{S}^{c} = (A \odot A)S^{c}(W \odot W)
\end{equation}
where, $\odot$ denotes element-wise multiplication.
\end{theorem}
\begin{proof}
Please refer to Appendix \ref{appendix1} for the proof.
\end{proof}

\vspace{-1.2em}

\subsection{Multi-head Graph Attention Network (MH-GAT)}
\label{MH-GAT}
Next, the GAT layer of GNN-SE is replaced by a MH-GAT in CGNN-SE.
The MH-GAT has $K$ attention heads running in parallel, with each head learning a different attention weight and hence, a different aspect of the relationship between the nodes.
For each head $k$, the attention coefficients $e^k_{ij}$ and normalized attention weights $\alpha^k_{ij}$ are calculated using \eqref{eq:GAT1} and \eqref{eq:GAT2}, but with head-specific parameters. Then, their corresponding feature vectors are concatenated, resulting in the following feature representation for each node $v_i$:
\begin{equation}\label{eq:GAT4}
x^{(l)}_i=\bigparallel_{k=1}^K \sigma \left(\sum_{j\in\mathcal{N}(i)\cup\{i\}}\alpha^{k,(l)}_{ij}W^{k, (l)}x^{(l-1)}_j \right)
\end{equation}
where, $k=1,2,\hdots, \mathrm{K}$ refers to different attention heads, $\bigparallel$ denotes the concatenation operation, and $\alpha^{k,(l)}_{ij}$ and $W^{k, (l)}$ are learnable attention weight and weight matrix associated with the $k^{th}$ attention head in the $l^{th}$ layer, respectively.
Note that, in this setting, the final node representation will have $Kf^{(l)}$ features rather than $f^{(l)}$ features, which will be adjusted by the linear transformation in the final linear layer of the CGNN-SE model.
An illustration of MH-GAT is shown in Fig. \ref{multi-head GAT}. 
Equation \eqref{eq:GAT4} can be written in matrix form to represent the whole graph as shown below:
\begin{equation}\label{eq:GAT5}
X^{(l)}=\bigparallel_{k=1}^K \sigma\left(\mathcal{A}^{k}X^{(l-1)}{W^\top}^{k, (l)}\right)
\end{equation}
where, $\mathcal{A}^k \in \R^{N \times N}$ is the attention weight matrix for the $k^{th}$ attention head whose elements are $\alpha^k_{ij}$ if $(v_i,v_j) \in \mathcal{E}$, and $\mathrm{0}$ otherwise. 
The concept of multi-head attention enhances the versatility of MH-GAT (over a regular GAT), resulting in improved feature representations for each node and enhanced estimation performance.

\begin{figure}[ht]
\centering
\includegraphics[width=0.394\textwidth]{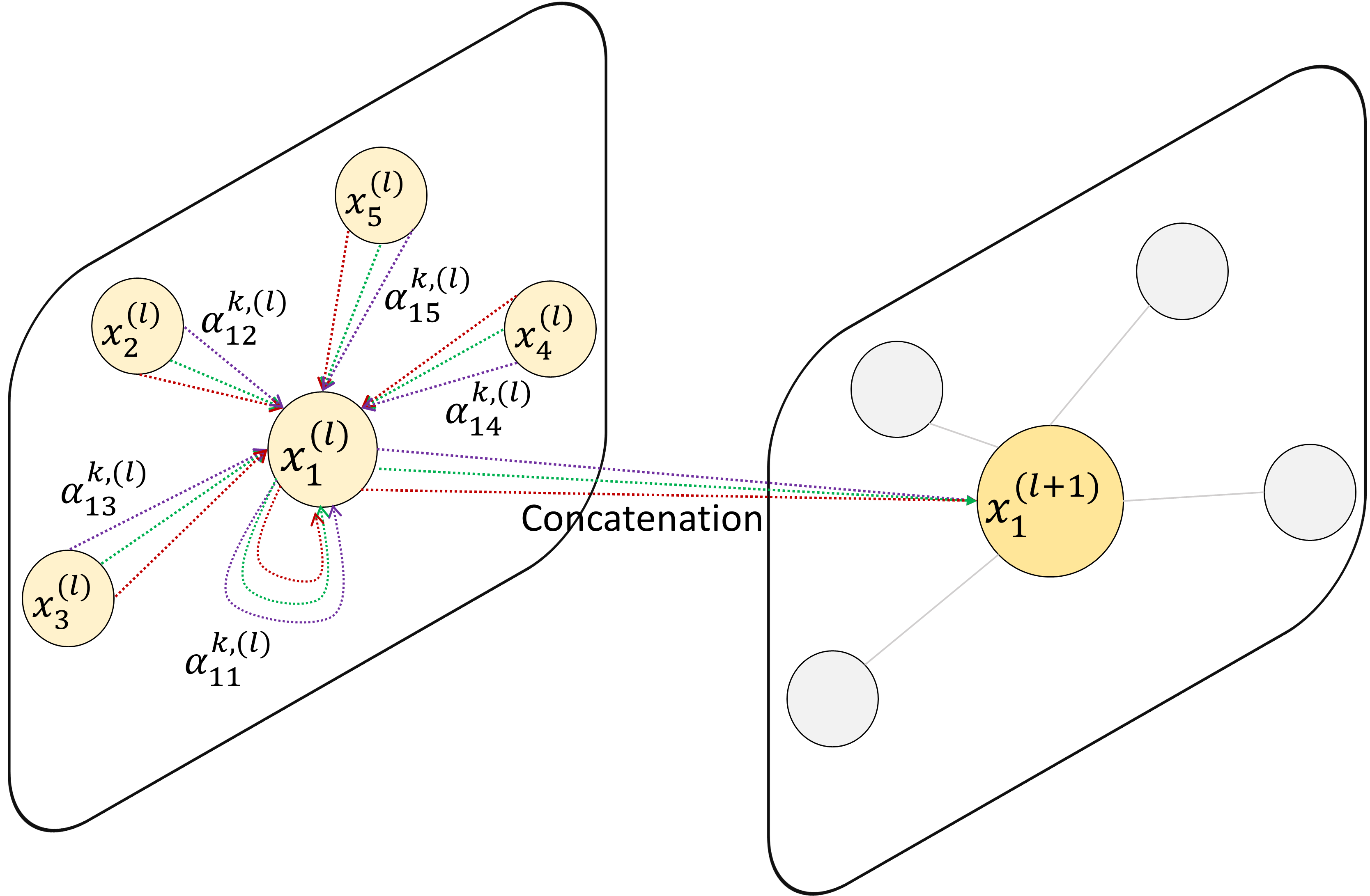}
\vspace{-1em}
\caption{An illustration of the MH-GAT for node 1 at layer $l$ with $K = 3$ heads. Various arrows represent distinct attention weights that capture diverse aspects of the relationships between neighboring nodes. The aggregated information is then concatenated to compute the node feature for node 1 in layer $l+1$.}
\label{multi-head GAT}
\end{figure}

\vspace{-1em}

\subsection{Explanation of Topology Adaptivity}
\label{explainability}
We now provide a mathematical explanation showcasing the stability of the proposed CGNN-SE to topology changes.
Let the mapping function modeled by the GNN of the CGNN-SE framework be $\Phi(A,X,W)$.
The topology change appears in the form of a perturbation to $A$, followed by a perturbation in the $X$ created using PMU data \cite{gama2020stability}. 
Let $A'$ and $X'$ be the perturbed adjacency and node feature matrices after a topology change.
In this context, Theorem \ref{theorem3} is defined to compare the output of the mapping function $\Phi(\cdot)$ under both scenarios (before and after the topology change) by evaluating the upper bound of the difference between the two outputs.

\begin{theorem}\label{theorem3}
Let $A$, $X$, $A'$, and $X'$ be the original and perturbed adjacency and node feature matrices, such that $\|A-A'\| \leq \epsilon$.
Given the mapping function $\Phi(A,X,W)$, if $\left| {w^{(l)}_{ij}}\right| \leq \delta$ and $\left\lVert w^{(l)}_{ij}A'\right\rVert \leq B $ 
where $w_{ij}^{(l)} \in W^{(l)}$, and $\lambda$ be the norm of first layer's features, then the following is true:
\begin{equation}\label{eq:exp2}
\lVert \Phi(A,X,W) - \Phi(A',X',W) \rVert \leq  \sqrt{2}\lambda\delta\epsilon L B^{L-1} F^{L-2}
\end{equation}
where, $L$ denotes the number of layers and $F$ denotes the number of features in each layer.
\end{theorem}

\begin{proof}
Please refer to Appendix \ref{appendix2} for the proof.
\end{proof}

Theorem \ref{theorem3} proves that in the event of a topology change resulting in a perturbed node feature matrix $X'$ and adjacency matrix $A'$,  
the difference between the outputs of the CGNN-SE framework is bounded by the expression in \eqref{eq:exp2}. 
Note that this bound depends on the architecture (number of layers and features) of the trained GNN of the CGNN-SE framework and the difference between the underlying and perturbed adjacency matrices.
Consequently, as the model becomes wider and deeper or this difference becomes larger, the bound becomes loose.
Therefore, during training, our aim is to identify the smallest architecture through hyperparameter tuning that gives a very low validation loss.
Doing this ensures that we have a balance between accuracy and stability.

\vspace{-1em}

\subsection{Overall Implementation}\label{overall-implementation}
The proposed CGNN-SE approach, illustrated in Fig. \ref{RGNN-SE} and detailed in Algorithm \ref{algorithm1}, uses \eqref{eq:GMM8} to analytically compute the expected activation of the nodes with missing features, represented as GMMs, in the first GNN layer.
This enables the processing of missing features and geometric learning within the same neural network architecture as an end-to-end procedure.
The expected activation is then employed as input for the next GNN layer within the CGNN-SE framework.
Finally, the GAT layer is replaced by a MH-GAT layer to better learn the intricate nodal relationships within the power system graph through multiple attention mechanisms.

\begin{figure*}
\centering
\includegraphics[width=0.9\textwidth]{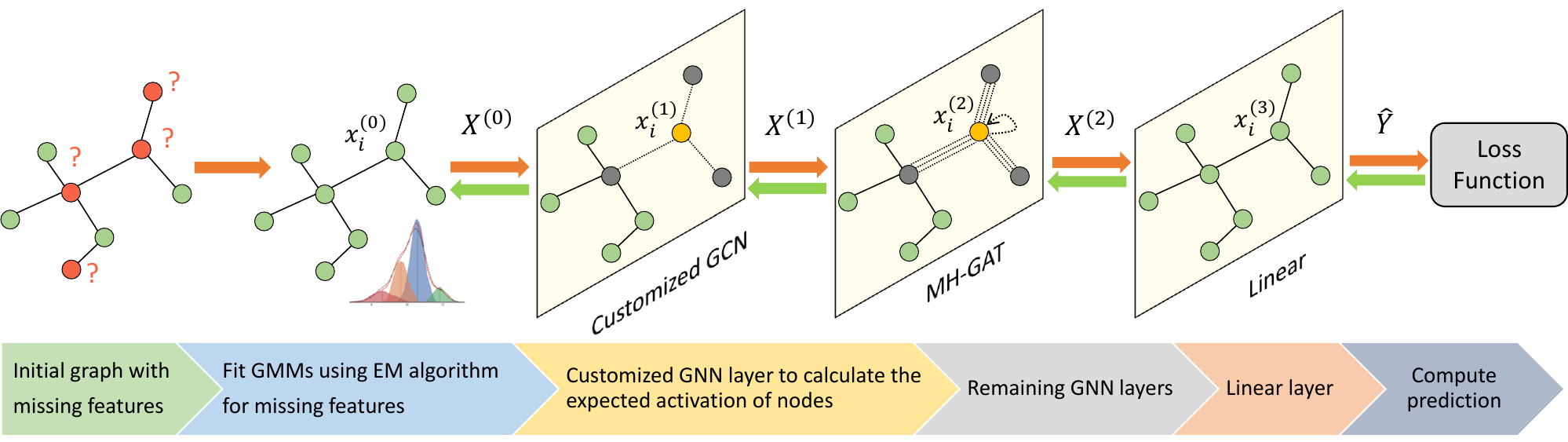}
\caption{The proposed CGNN-SE framework showcases an end-to-end geometric learning process, where a power system graph with missing features is fed into a customized GCN layer to calculate expected activation of the nodes. The diagram also illustrates the message-passing process by which a given node (highlighted in yellow in the figure) extracts information from its neighbors.
}
\label{RGNN-SE}
\end{figure*}

The model parameters described in Algorithm \ref{algorithm1} include GNN parameters (weights and biases) as well as GMM parameters (weight, mean, and covariance). 
The algorithm simultaneously learns both types of parameters to enhance the overall learning process. Note that GNN parameters are initialized randomly, while GMM parameters are initialized using EM 
to ensure a favorable starting point.
More details about the initialization process is provided in Section \ref{simulation-setup}.

\begin{algorithm}
	\caption{Algorithm of the proposed CGNN-SE}\label{algorithm1}
 \textbf{Input:} Default connections of the power system to form the adjacency matrix $A$; Locations of existing PMUs; PMU data for nodes equipped with PMUs; Historical SCADA data for all the nodes; 
 Number of Gaussian components $C$ to form the node feature matrix $X$; Number of GNN layers $L$, Number of attention heads $K$.\\
  \textbf{Output:} Predicted states (voltage magnitude and phase angle) of all nodes in the power system graph 

\begin{algorithmic}[1]
\STATE \textbf{Initialize} model parameters including GNN parameters and GMM parameters 
\WHILE{$\mathrm{epoch} < \mathrm{epochs}$}
\STATE $X^{(1)} \leftarrow$ Eq. \eqref{eq:GMM8}
\FOR{$l=2, \dots, L-1$}
\STATE $X^{(l)} \leftarrow$ Eq. \eqref{eq:GCN2}
\ENDFOR
\STATE $X^{(L)} \leftarrow$ Eq. \eqref{eq:GAT5}
\STATE $\hat{Y} \leftarrow $ Eq. \eqref{eq:linear}
\STATE $\mathcal{L} \leftarrow \mathcal{L}(\hat{Y})$
\STATE Minimize $\mathcal{L}$ and update GNN parameters and GMM parameters using gradient descent
\ENDWHILE

\end{algorithmic}
\end{algorithm}

Once the proposed CGNN-SE model is trained, it can be employed in real-time as measurements arrive from PMUs. Given that PMU measurements may contain bad data in the form of outliers that can negatively impact the performance of the CGNN-SE, a bad data detection and correction (BDDC) scheme is necessary.
For performing ML-based SE in PMU-unobservable power systems, \cite{mestav2019bayesian} proposed a Wald test-based BDDC scheme to tackle outliers.
The Wald test detects bad data by identifying unexpected or significantly deviant data points in the statistical sense
\cite{liu2018multichannel}.
This test is based on two hypotheses: the null hypothesis (\(H_0\)), which represents the distribution of measurements \textit{without} bad data, characterized by a mean \(\mu_0\) and variance \(\sigma_0^2\) that are learned during training, and the alternative hypothesis (\(H_1\)), which represents the distribution of measurements \textit{with} bad data, featuring a mean and variance that differ significantly from
those of 
\(H_0\).
Mathematically, the Wald test can be expressed as:
\begin{equation}\label{eq:wald_test}
    \left |\frac{z-\mu_0}{\sigma}\right| \lessgtr_{H_0}^{H_1} Q^{-1}\left(\frac{\alpha}{2}\right)
    \end{equation}
In \eqref{eq:wald_test}, $Q(y) = \frac{1}{\sqrt{2\pi}} \int_{y}^{\infty} \exp\left(-\frac{u^2}{2}\right)du$ denotes the tail of the distribution, and $\alpha$ is a tunable parameter that specifies the false positive rate.
The Wald test leverages the fact that training is performed using good-quality data\footnote{This is ensured by \textit{not} training the CGNN-SE directly using the SCADA data since it may be erroneous; the training details are provided in Section \ref{simulation-setup}.}. Once the limits of good quality data are established during training, any testing data that falls outside these limits can be classified as bad data. Identified bad data points are then corrected by replacing them with the mean value obtained from the training dataset.
We selected the Wald test as our BDDC scheme because it is well-suited for ML-based SE techniques and can operate at PMU timescales ($\leq 33$ ms).
In the proposed framework, this scheme is employed as a pre-processing step during the online stage to minimize the impact of bad data.


\section{Simulation Results}\label{results}
In this section, we evaluate the performance of the proposed CGNN-SE in comparison to conventional optimization-based and learning-based state estimators for the IEEE 118-bus and 2000-bus Texas systems, as well as with a state-of-the-art static state estimator (\cite{jabr2023complex}) for the 1354pegase system (part of the European high-voltage transmission network from the PEGASE project).
The robustness of the CGNN-SE is also verified under diverse scenarios, including limited PMU coverage, topology changes, real-time PMU failures, and presence of non-Gaussian measurement noise and bad data.

\subsection{Simulation Setup and CGNN-SE Structure}
\label{simulation-setup}
The training process for the CGNN-SE involves using many training graph samples to represent different operating conditions (OCs).
The training graphs are defined using adjacency and node feature matrices. 
The adjacency matrix remains fixed for all training graphs, but the 
entries of the node feature matrices change as they contain 
PMU data from different OCs.
The overall process of data generation and graph definition, which occur during the offline stage, is depicted in Fig. \ref{flowchart}. It begins with historical load data from SCADA (which may be erroneous) being fed into a distribution learner, where the distribution of all loads is learnt using Kernel density estimation.
Independently picked samples from these distributions are set as inputs to a power flow solver to generate voltage phasors for different OCs. The outputs of the solver are directly used for nodes equipped with PMUs, after adding appropriate noise (e.g., 1\% total vector error (TVE) \cite{8577045}).
For a node without PMU (missing features),
a GMM is fitted to the output of the power flow solver.
By not using the SCADA data directly to train the CGNN-SE, we ensure that any reasonable errors present in the SCADA data do not impact the proposed state estimator's performance.

\begin{figure*}[ht]
\centering
\includegraphics[width=0.68\textwidth]{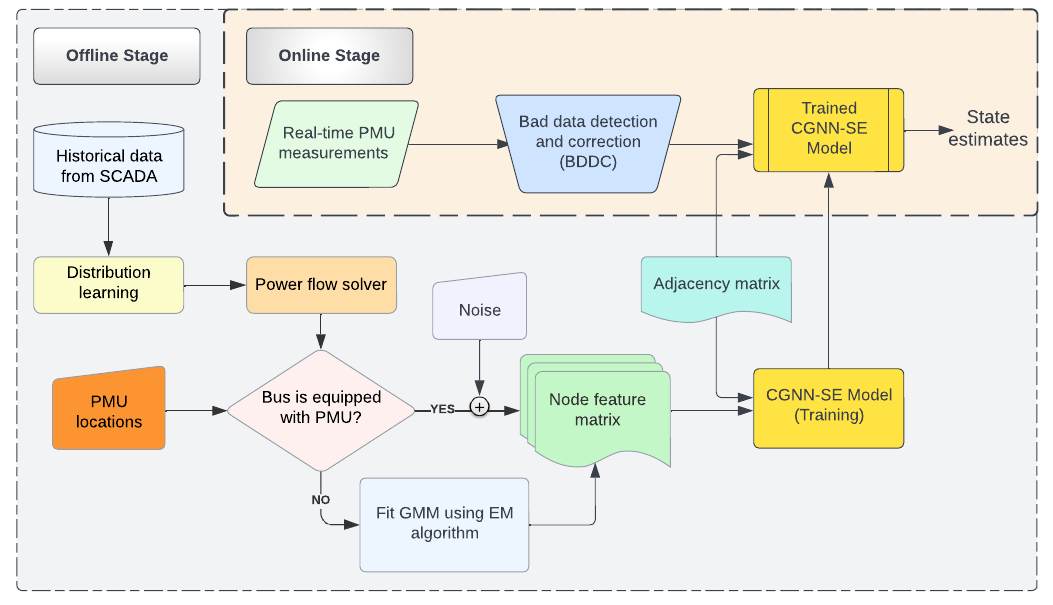}
\vspace{-1.2em}
\caption{Overall implementation of the proposed CGNN-SE}
\label{flowchart}
\end{figure*}

Note that unless specified otherwise, it is assumed that PMUs are placed only on the highest voltage buses (11 buses at 345 kV level for the 118-bus system and 120 buses at 500 kV level for the 2000-bus system).
This is based on the fact that power utilities first place PMUs at the highest voltage level.
Moreover, when a PMU is placed at a bus, it is assumed that all outgoing lines of that bus are observed. This is because the additional cost of placing more PMUs at buses that already have at least one PMU, is incremental \cite{pal2017general}.
For the 1354pegase system, the number of PMUs and noise model are chosen based on \cite{jabr2023complex}.

The CGNN-SE architecture consists of one GCN layer followed by a MH-GAT layer, with $K=4$ attention heads, with both layers having $f^{(l)}=50$ hidden features.
Note that the first layer is customized to calculate the expected activation of nodes with missing features represented as GMMs to handle the problem of PMU scarcity as explained in Section \ref{unobservability}.
The CGNN-SE model was trained using 5000 graph samples with a batch size of 10 on the 118-bus system, resulting in a training time of approximately 20 seconds/epoch. For the larger 2000-bus system, the same architecture and batch size of 2 resulted in a training time of around 2 minutes/epoch.

The online stage of the proposed framework, shown in Fig. \ref{flowchart}, is the real-time processing phase in which PMU measurements, after passing through the BDDC block,
are fed into a pre-trained CGNN-SE.
All simulations were conducted on a high-performance computer equipped with 256 GB of RAM, an Intel Xeon 6246R CPU operating at 3.40 GHz, and an Nvidia Quadro RTX 5000 with 16 GB of GPU memory.
Depending on the complexity of the system and number of epochs, the entire training process can take several days to complete.
During the training process, we employed early stopping to prevent overfitting and reduce training time, with the model typically converging within 4000 epochs.
During the training, the GPU showed a maximum memory usage of approximately 11.8\% (1955 MB out of 16384 MB).
Lastly, note that training is done in the offline stage, which will not impact the online (real-time) operation of the CGNN-SE. Furthermore, the inherent robustness of the CGNN-SE to topology changes implies that the training need not be repeated often (unlike in the case of a DNN-SE).

\subsection{State Estimation Results for 118-bus and 2000-bus Systems}\label{feasibility}
We verify the robustness of the proposed CGNN-SE framework to missing features in PMU-unobservable power systems by comparing its performance with other optimization-based and learning-based benchmarks.
The optimization-based state estimators include SCADA-SE and LSE, both of which obtain estimates by solving the maximum likelihood estimation problem in the least squares sense \cite{zhao2020robust}.
Measurements for SCADA-SE comprise all sending-end active power flows and voltage magnitudes, corrupted by 20\% and 10\% additive Gaussian noise, respectively. 
The LSE results are obtained by placing PMUs at 32, and 512 optimal locations, which are the minimum needed for full system observability in 118-bus and 2000-bus systems, respectively, \cite{pal2013pmu}. 
PMU measurements are subjected to 1\% Gaussian TVE for this case-study.

The learning-based state estimators include support vector regression (SVR), DNN-SE, GNN-SE, and the proposed CGNN-SE.
The SVR, DNN-SE and CGNN-SE are trained and tested using PMU measurements coming from the highest voltage buses, while the GNN-SE receives PMU measurements from the optimal locations used in LSE, to satisfy its fully-observed node feature matrix requirement.
The DNN in DNN-SE has a fully-connected feed-forward architecture \cite{varghese2024deep}, while the architecture of the GNN for GNN-SE was obtained from \cite{10318579}.
The optimization-based state estimators were implemented using MATLAB, whereas TensorFlow was used for DNN-SE, and PyTorch Geometric was used for GNN-SE and CGNN-SE.
A comparison of the results of these six state estimators for this case-study is provided in Table \ref{table1}.

\begin{table}[ht]
\centering
\caption{Performance comparison of the Proposed CGNN-SE model with optimization-based and learning-based state estimators}
\resizebox{\columnwidth}{!}{
\begin{tabular}{@{}clcccc@{}}
\toprule
\multirow{2}{*}{\rotatebox[origin=c]{90}{\textbf{System}}} & \multirow{2}{*}{\textbf{Approach}} & \multirow{2}{*}{\textbf{\begin{tabular}[c]{@{}c@{}}Num. of buses\\ with PMUs\end{tabular}}} & \multirow{2}{*}{\textbf{\begin{tabular}[c]{@{}c@{}}Magnitude\\ MAPE (\%)\end{tabular}}} & \multirow{2}{*}{\textbf{\begin{tabular}[c]{@{}c@{}}Phase Angle\\ MAE (degrees)\end{tabular}}} & \multirow{2}{*}{\textbf{\begin{tabular}[c]{@{}c@{}}Computation\\ Time (s)\end{tabular}}} \\
& & & & & \\ [2pt]
\midrule
\multirow{6}{*}{\rotatebox[origin=c]{90}{IEEE 118-bus}} & SCADA-SE & 0 & 0.746 & 0.306 & 1.24 \\
& LSE & 32 & 0.270 & 0.143 & 8.2e-3 \\
& SVR & 11 & 0.309 & 0.698 & 8.1e-4 \\
& DNN-SE & 11 & 0.161 & 0.230 & 1.0e-4 \\
& GNN-SE & 32 & 0.047 & 0.080 & 7.2e-4 \\
& \textbf{CGNN-SE} & \textbf{11} & \textbf{0.018} & \textbf{0.027} & \textbf{9.0e-4} \\
\midrule 
\multirow{6}{*}{\rotatebox[origin=c]{90}{2000-bus Texas}} & SCADA-SE & 0 & 0.654 &  2.121 & 133 \\
& LSE & 512 & 0.261 & 0.138 & 1.27 \\
& SVR & 120 & 0.664 & 0.502 & 7.8e-2 \\
& DNN-SE & 120 & 0.180 & 0.097 & 2.0e-4 \\
& GNN-SE & 512 & 0.151 & 0.083 & 3.3e-4 \\
& \textbf{CGNN-SE} & \textbf{120} & \textbf{0.093} & \textbf{0.064} & \textbf{1.9e-2} \\
\bottomrule
\end{tabular}%
}
\label{table1}
\end{table}

Table \ref{table1} displays the mean absolute percentage error (MAPE) for magnitude estimates and the mean absolute error (MAE) for phase angle estimates.
It is evident from the table that the proposed CGNN-SE outperforms both optimization-based and learning-based state estimators with fewer/similar number of PMUs.
The exceptional performance of the CGNN-SE is attributed to its ability to capture intricate relationships among the nodes using customized GCN and MH-GAT layers, as well as its ability to jointly optimize network parameters and GMM parameters within the end-to-end geometric learning process.

The average computation time for producing an estimate for each state estimator is also listed in Table \ref{table1}.
It is clear from the entries that the learning-based models consistently outpace the optimization-based models.
Note that GNN-SE and CGNN-SE have slightly higher computation times compared to DNN-SE since they analyze both the data and the structure of the power system.
Additionally, the slightly higher computation time for CGNN-SE compared to GNN-SE (particularly, for the 2000-bus system) can be attributed to its processing of (many) missing features in the form of GMM parameters. 
Lastly, it can be inferred from the table that even for large systems, CGNN-SE is able to operate at PMU timescales ($\leq$ 33 ms), which is a requirement for time-synchronized PMU-only SE.

\subsection{Robustness to Topology Changes}\label{topology-change}
Next, we tested the robustness of CGNN-SE in handling topology changes during online evaluation.
Such changes are common in practical power system operations \cite{huang2015dynamic,li2017real}, and typically require retraining conventional learning-based SE models such as DNN-SE to update mapping relationships for the new topology.
However, geometric learning-based techniques, such as CGNN-SE, offer robust estimation without the need for \textit{ex post facto} learning after topology change.
To verify this, we evaluated the performance of the \emph{pre-trained} CGNN-SE and DNN-SE for the nominal (base) topology under five different cases of single line openings. 
The outage was done on those lines that had the highest power flowing through them, but did not create islands in the system after disconnection. 
The results are depicted in Fig. \ref{topology} in terms of magnitude MAPE and phase angle MAE. 
It is evident from Fig. \ref{topology} that topology changes during online evaluation pose challenges for the generalization capability of DNN-SE, while CGNN-SE exhibits immunity to such changes owing to its topology-aware structure (bars in lighter shades are much taller than the bars in darker shades).

\begin{figure}[ht]
\centering
\includegraphics[width=0.47\textwidth]{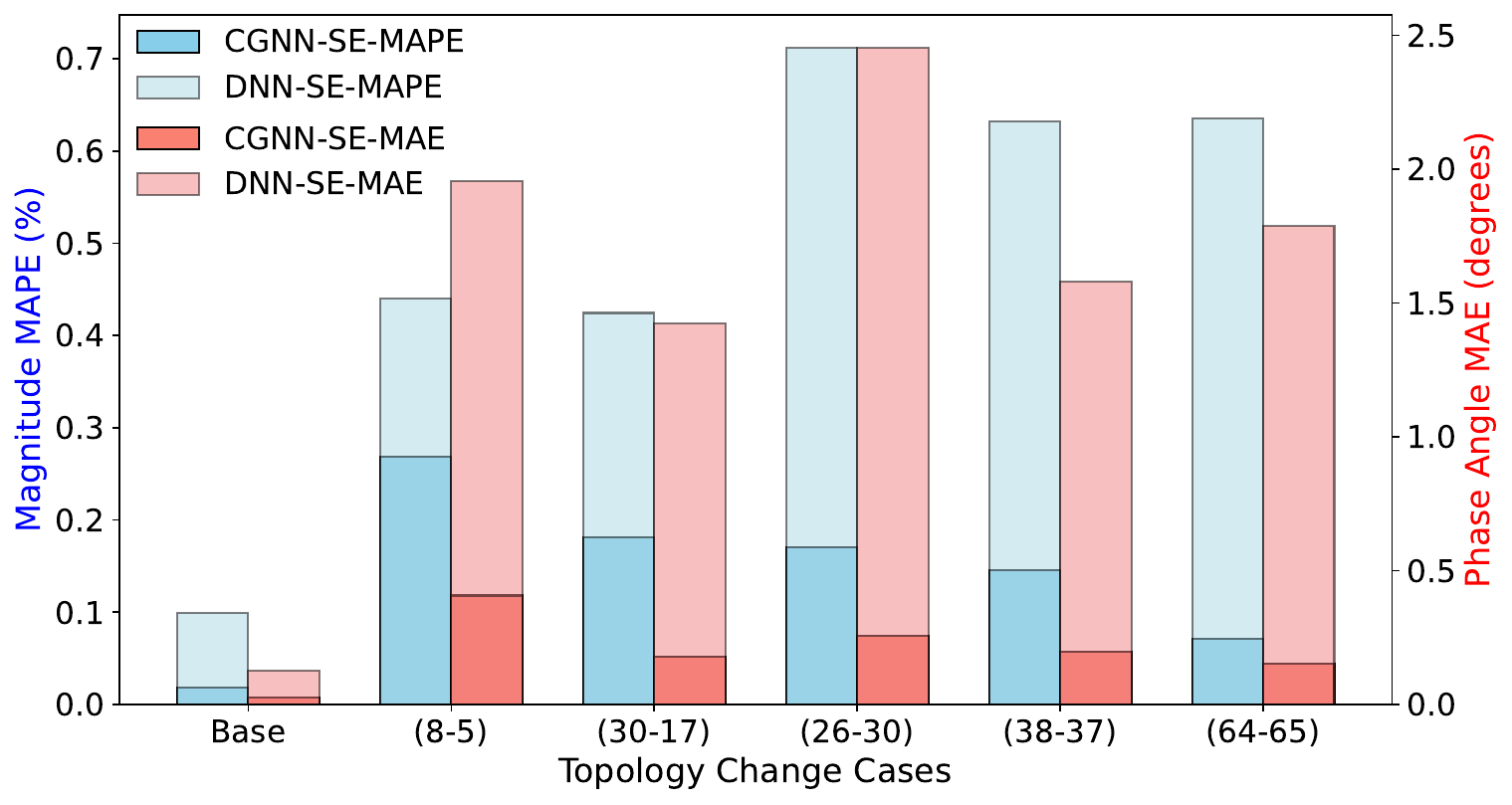}
\caption{Comparing magnitude MAPE and phase angle MAE of DNN-SE and CGNN-SE for outages of five lines of the IEEE 118-bus system that have the highest power flowing through them. On the X-axis, $\mathrm{Base}$ refers to the nominal (base) topology where there are no line outages. The other indices denote the \textit{from-to} buses of the lines that have been removed.
}
\label{topology}
\end{figure}

The robustness of the proposed CGNN-SE may vary depending on which specific line opens.
This is because it is a function of the proximity of the line to the PMUs
as well as the resulting shift in the power system’s overall operating point.
We recognized that outages of lines with higher power flows have a greater effect on the system states \cite{10318579}.
Therefore, the scenarios considered in Fig. \ref{topology} could be deemed as the worst-case single line outage ($N-1$) conditions that the CGNN-SE model could be subjected to.

The topology adaptivity of the proposed CGNN-SE was explained mathematically in Section \ref{explainability} by providing an \textit{upper bound} on the change in SE output after a topology change.
Note that the upper bound provided in \eqref{eq:exp2} holds for $N-k$ contingencies, where $k\geq1$. 
This upper bound depends on the structure of the CGNN-SE and the difference between the original and perturbed adjacency matrices. 
While the former can be controlled by hyperparameter tuning, the latter is a function of the severity of the topology change (i.e., the number of line openings in the power system graph).
Considering the current structure of the CGNN-SE, we verified its stability in Fig. \ref{upper-bound} under multiple line openings. The figure displays the magnitude MAPE and phase angle MAE of the CGNN-SE for four scenarios, where $\mathrm{Base}$ represents the nominal topology without any line outages, while $N-1$, $N-2$, and $N-3$ represent one, two, and three line openings, respectively.
It can be realized from the figure
that as more lines open, the performance of DNN-SE deteriorates significantly (there is a marked increase in both magnitude and phase angle errors).
However, CGNN-SE remains relatively stable with low magnitude and angle errors, demonstrating its robustness and superior adaptability to varying network topologies. 

\begin{figure}[t]
\centering
\includegraphics[width=0.489\textwidth]{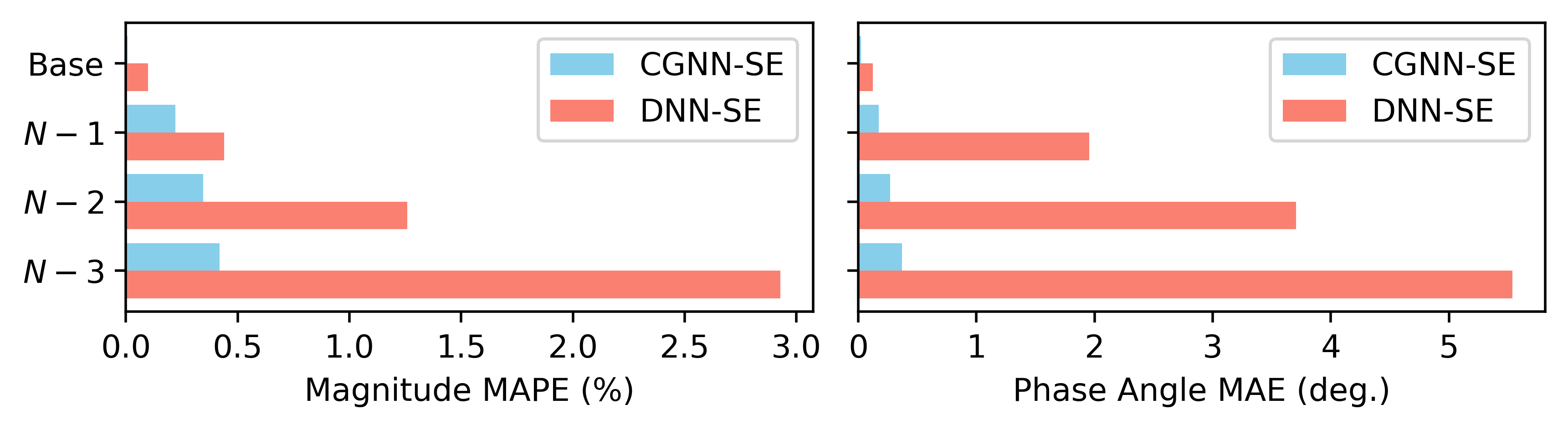}
\caption{Stability of the proposed CGNN-SE for IEEE 118-bus system under multiple line openings. ${N-1}$ corresponded to opening of line between \{(8-5)\}, ${N-2}$ corresponded to opening of lines between \{(8-5), (30-17)\}, while ${N-3}$ corresponded to opening of lines between \{(8-5), (30-17), (38-37)\}.}
\label{upper-bound}
\end{figure}

\begin{figure}[b]
\centering
\includegraphics[width=0.48\textwidth]{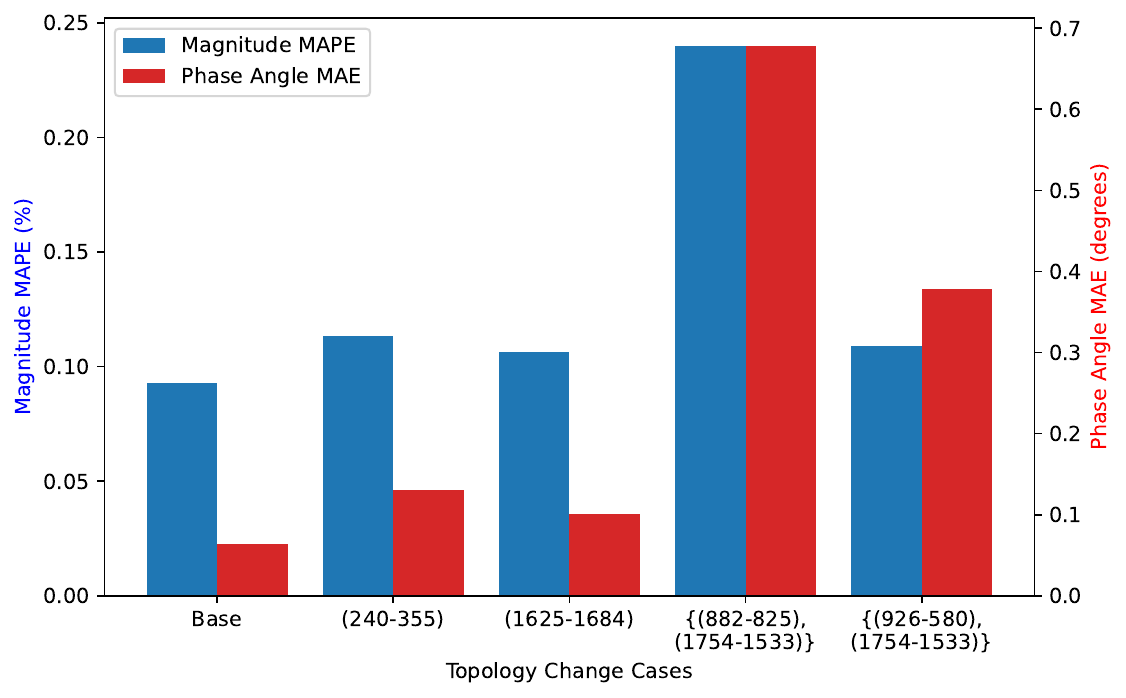}
\vspace{-1em}
\caption{Robustness of CGNN-SE to topology changes including $N-1$ and $N-2$ contingencies in 2000-bus Texas system. On the X-axis, $\mathrm{Base}$ refers to the nominal (base) topology where there are no line outages. The other indices denote the \textit{from-to} buses of the lines that have been removed.}
\label{topology_change_2000-bus}
\end{figure}

To assess the robustness of the proposed CGNN-SE to topology changes in large-scale systems, we tested its performance on the 2000-bus system under $N-1$ and $N-2$ contingencies. The results of this evaluation are shown in Fig. \ref{topology_change_2000-bus}, which highlights the robustness of the CGNN-SE model for different scenarios, in comparison to the $\mathrm{Base}$ case with no line outages.
As illustrated in the figure, the CGNN-SE maintains strong performance across these scenarios. While the cases involving two line outages exhibit relatively higher estimation errors than those with single line outages, the results confirm
the reasonableness of the CGNN-SE outputs under such conditions.

\subsection{Robustness to PMU Failure}
We now illustrate the robustness of the CGNN-SE in scenarios involving real-time PMU failures, where data loss occurs during online evaluation. 
Fig. \ref{pmu-failure} depicts the \textit{average} MAPE and MAE for magnitudes and phase angles, respectively, for different combinations of PMU failures
in the 118-bus system. For example, the third bar corresponds to the average estimation errors of ${\binom{11}{2}}$ cases, where PMUs fail at any two buses (out of 11 buses) in this system.
It is observed from the figure that after the initial rise (from when there is no PMU failure vs. one PMU failure), the errors in magnitude and angle estimates almost plateau out.
Similar results were also obtained for the 2000-bus system. 
This robustness (to PMU failures) is attributed to the CGNN-SE's distinctive structure that allows it to substitute missing features with their corresponding GMM parameters during online evaluation.

\begin{figure}[ht]
\centering
\includegraphics[width=0.45\textwidth]{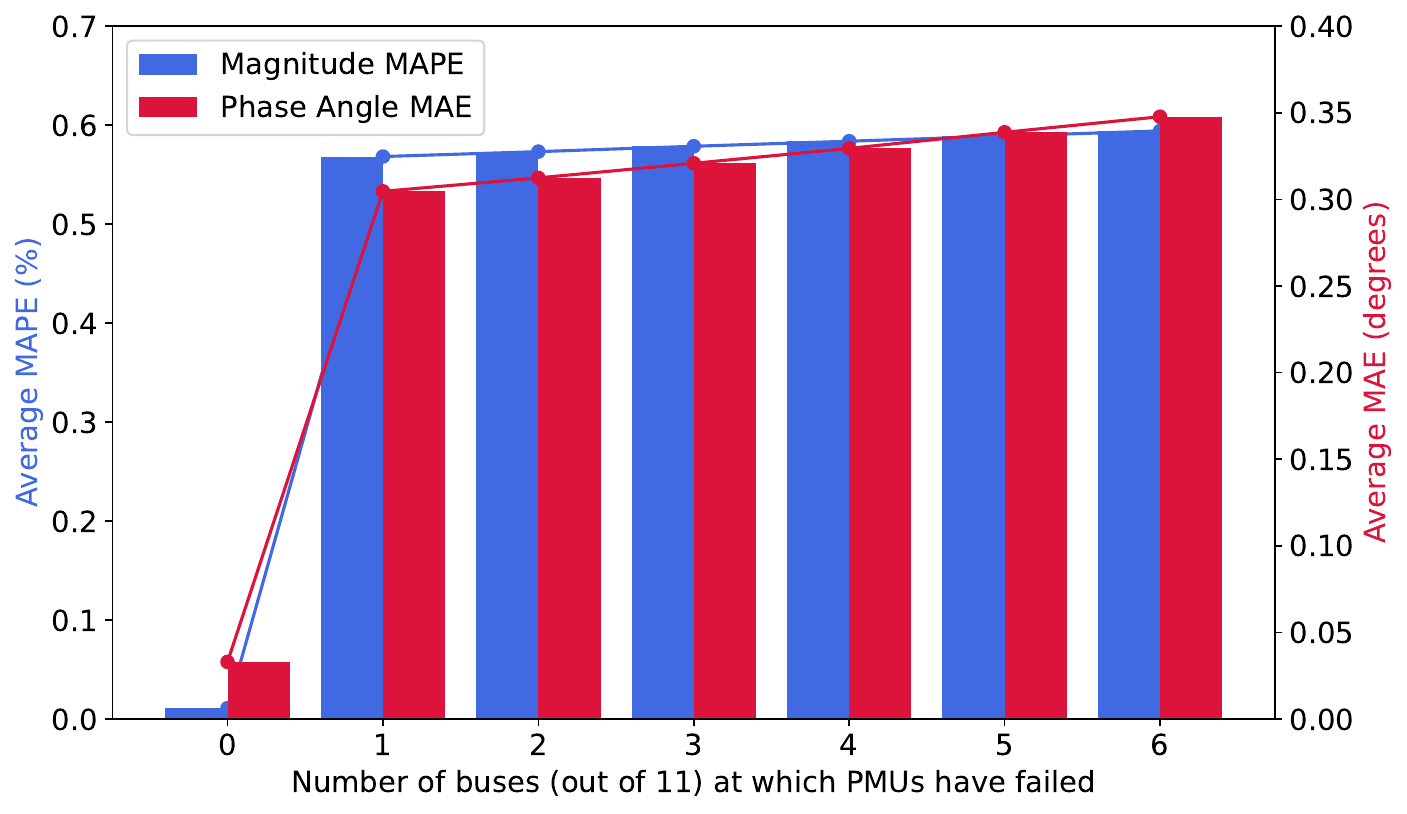}
\vspace{-1em}
\caption{Robustness of CGNN-SE to PMU failures in IEEE 118-bus system}
\label{pmu-failure}
\end{figure}

\subsection{Robustness to Topology Changes and PMU Failures}
To further verify the robustness of the proposed CGNN-SE, we conducted experiments under scenarios where both topology changes and PMU failures occur simultaneously. 
This corresponds to a scenario where there are structural changes occurring in both the power system \textit{and} the sensing system. 
The aim was to assess how the CGNN-SE model responds under such adverse conditions.
The results of this study are illustrated for the 118-bus system in Fig. \ref{topology_change_pmu-failure}, which compares performance across four different cases alongside the $\mathrm{Base}$ case, where no topology change or PMU failure occurs.
It can be observed from the figure that the CGNN-SE model continues to give reasonable
performance even when subjected to simultaneous topology changes and PMU failures.

\begin{figure}[h]
\centering
\includegraphics[width=0.48\textwidth]{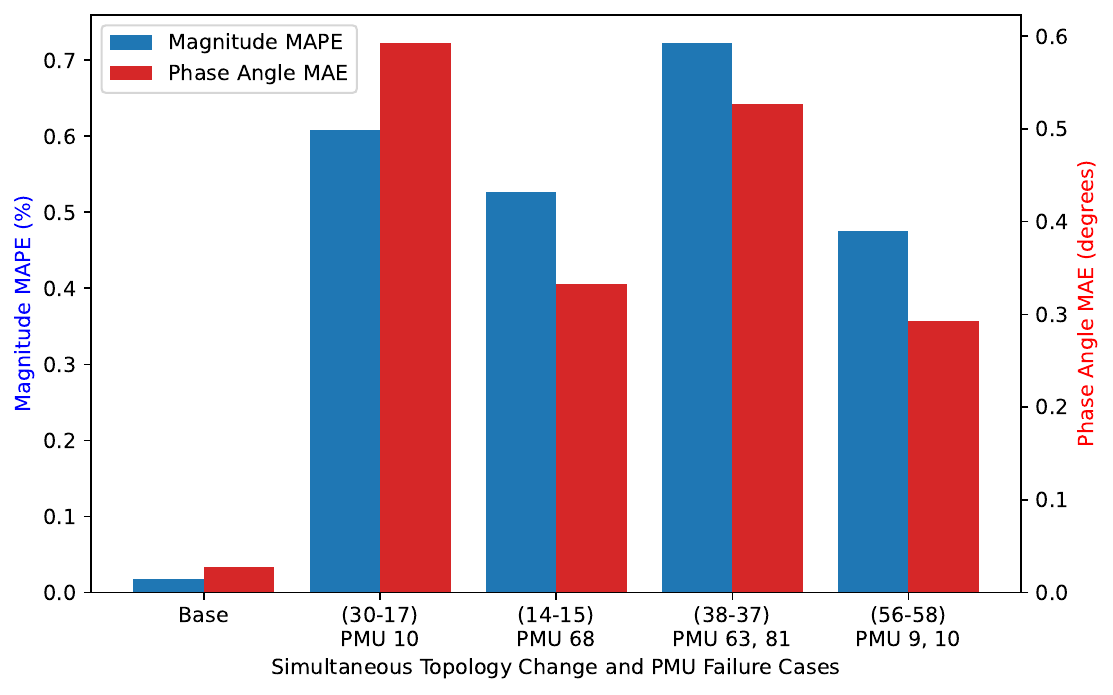}
\vspace{-0.5em}
\caption{Robustness of CGNN-SE to simultaneous topology change and PMU failure in IEEE 118-bus system. On the X-axis, $\mathrm{Base}$ refers to the nominal (base) topology with no line outages or PMU failures. The remaining indices represent the removed line (\textit{from-to} buses) and the corresponding PMU failure(s). For example, ``(30-17) PMU 10" indicates a line outage between buses 30 and 17 with a PMU failure at bus 10.}
\label{topology_change_pmu-failure}
\end{figure}

\vspace{-1em}

\subsection{Robustness to Non-Gaussian Noise in PMU Measurements}
Table \ref{table2} compares the performance of optimization-based and learning-based state estimators in presence of Gaussian and non-Gaussian noises in PMU measurements. 
The non-Gaussian noise was described by a two-component GMM
whose mean, standard deviation, and weights for both magnitudes and phase angles are $(-0.4\%, 0.6\%)$ and $(-0.2^\circ, 0.3^\circ)$, $(0.25\%, 0.25\%)$ and $(0.12^\circ, 0.12^\circ)$, and $(0.4, 0.6)$, respectively.
These values were chosen to ensure that the noise still satisfied the 1\% TVE requirement of \cite{8577045}.
It is clear from Table \ref{table2} that the least squares-based LSE is not robust against non-Gaussian noise as the errors increased considerably.
Due to their better approximation capabilities, the learning-based estimators are more robust against non-Gaussian measurement noises with the CGNN-SE showing best performance (lowest error) because of its specific structure and customized geometric learning process.

\vspace{-0.5em}

\begin{table}[ht]
\caption{Robustness of CGNN-SE to non-Gaussian noise in PMU measurements on IEEE 118-bus system}
\resizebox{\columnwidth}{!}{%
\begin{tabular}{@{}lcccc@{}}
\toprule
\multirow{3}{*}{\textbf{Approach}} & \multicolumn{2}{c}{\textbf{1\% TVE Gaussian Noise}} & \multicolumn{2}{c}{\textbf{1\% TVE Non-Gaussian Noise}} \\ 
\cmidrule(lr){2-3} \cmidrule(lr){4-5} 
 & \multirow{2}{*}{\begin{tabular}[c]{@{}c@{}}Magnitude\\ MAPE (\%)\end{tabular}} & \multirow{2}{*}{\begin{tabular}[c]{@{}c@{}}Phase Angle\\ MAE (degrees)\end{tabular}} & \multirow{2}{*}{\begin{tabular}[c]{@{}c@{}}Magnitude\\ MAPE (\%)\end{tabular}} & \multirow{2}{*}{\begin{tabular}[c]{@{}c@{}}Phase Angle\\ MAE (degrees)\end{tabular}} \\
 &  &  &  &  \\ 
\midrule
LSE & 0.270 & 0.143 & 0.492 & 0.246 \\
DNN-SE & 0.161 & 0.230 & 0.173 & 0.250 \\
GNN-SE & 0.047 & 0.080 & 0.054 & 0.095 \\
\textbf{CGNN-SE} & \textbf{0.018} & \textbf{0.027} & \textbf{0.022} & \textbf{0.035} \\
\bottomrule
\end{tabular}%
}
\label{table2}
\end{table}

\subsection{Robustness to Bad Data}
In addition to non-Gaussian noise, bad data in the form of outliers also deteriorate performance of power system state estimators. 
To demonstrate this, we injected bad data, ranging from 10\% to 50\% of the input features (PMU measurements), directly into the CGNN-SE model for the 118-bus system, and observed a significant degradation in its performance (plots with dots in Fig. \ref{bad-data}).
Next, we employed the Wald test-based BDDC scheme (see Section \ref{overall-implementation} and Fig. \ref{flowchart}) to detect and correct bad data \textit{before} it entered the proposed state estimator. The state estimator outputs after bad data correction are also shown in Fig. \ref{bad-data} (plots with cross).
The figure shows that the estimation errors increase considerably if bad data is left untreated. 
However, the Wald test-based BDDC scheme is able to successfully detect and correct the bad data and keep the errors low for the proposed CGNN-SE framework.

\begin{figure}[ht]
\centering
\includegraphics[width=0.45\textwidth]{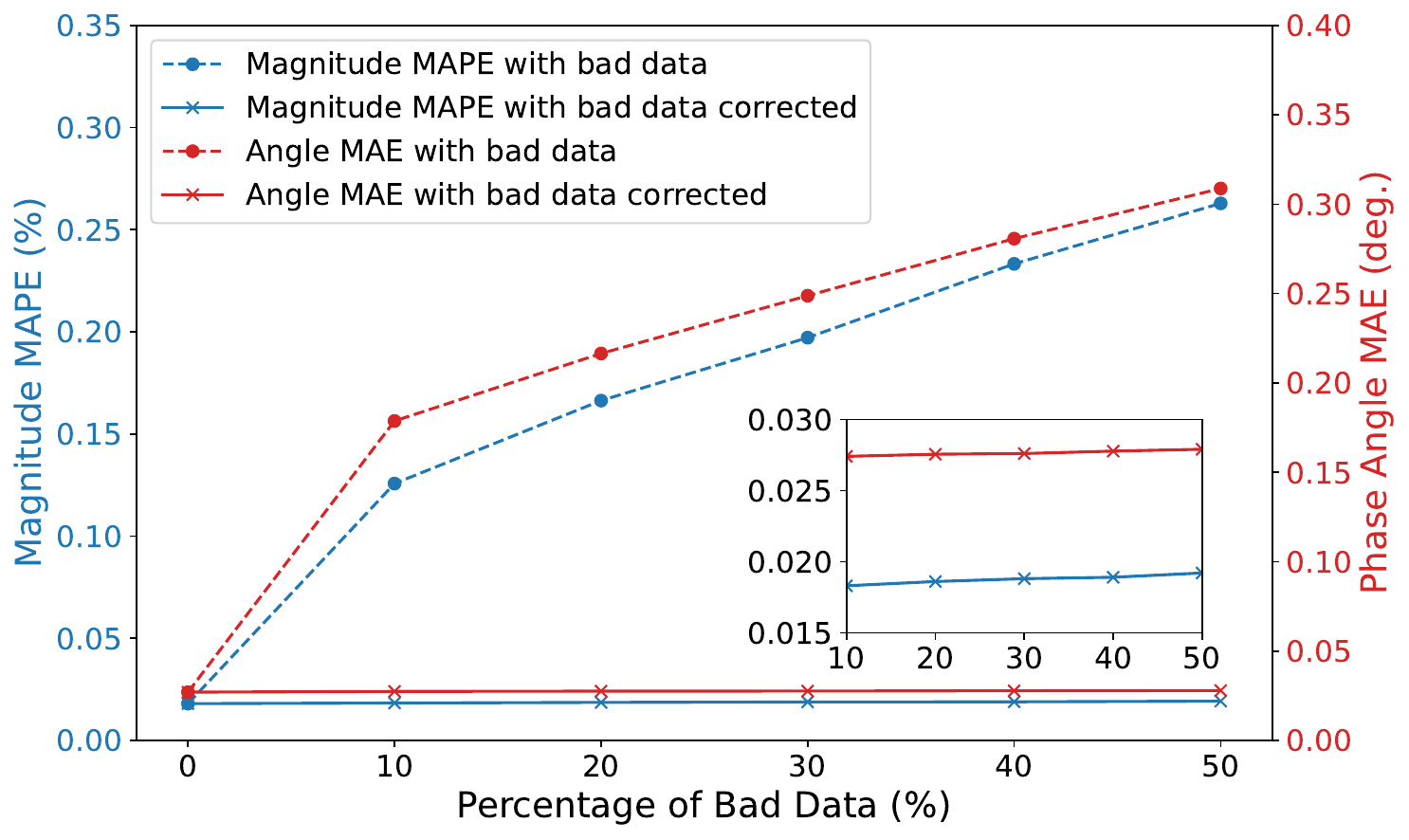}
\vspace{-1em}
\caption{Performance of CGNN-SE without and with bad data correction in IEEE 118-bus system. The inset provides a more zoomed-in view.
}
\label{bad-data}
\end{figure}

\subsection{Significance of Attention Mechanism}\label{attention}
To rigorously evaluate the effectiveness of the attention mechanism in the form of the MH-GAT layer within the proposed CGNN-SE model, we conducted an ablation study.
This involved comparing the performance of the proposed CGNN-SE model with and without the MH-GAT layer on the 118-bus system.
The results of this study are presented in Table \ref{table_ablation1}.
As demonstrated by the results, the exclusion of the MH-GAT layer significantly increased estimation errors (by a factor of $10$), confirming its importance in enhancing the model’s ability to prioritize relevant spatial features and thereby, improve estimation performance.

\begin{table}[ht]
\centering
\caption{CGNN-SE performance with and without MH-GAT layer on IEEE 118-bus system}
\resizebox{0.9\columnwidth}{!}{
\renewcommand{\arraystretch}{1.3} 
\begin{tabular}{@{}lcc@{}}
\toprule 
\multirow{2}{*}{\textbf{Approach}}  & \multirow{2}{*}{\textbf{\begin{tabular}[c]{@{}c@{}}Magnitude\\ MAPE (\%)\end{tabular}}}  & \multirow{2}{*}{\textbf{\begin{tabular}[c]{@{}c@{}}Phase Angle\\ MAE (degrees)\end{tabular}}}  \\
  \\
\midrule 
\textbf{CGNN-SE (with MH-GAT)} & \textbf{0.018} & \textbf{0.027} \\ 
CGNN-SE (without MH-GAT) & 0.179 & 0.314\\ 
\bottomrule
\end{tabular}
}
\label{table_ablation1}
\end{table}

Additionally, we explored the effect of varying the number of heads $K$ in the MH-GAT layer.
To assess the impact of this choice, we trained multiple CGNN-SE models with different numbers of heads on the 118-bus system.
The results, shown in Table \ref{table_ablation2}, indicate that increasing the number of heads initially led to reductions in estimation errors, suggesting that multiple heads contribute to more robust feature extraction by allowing the model to attend to different aspects of the data simultaneously. However, as is common in MH-GAT models, beyond a certain point, adding more heads increases the model’s complexity, which can lead to overfitting and decreased performance—a phenomenon known as diminishing returns 
\cite{ye2021sparse,wang2020multi,gong2024subgraph}.
In this case, as shown in Table \ref{table_ablation2}, increasing the number of heads from $K=1$ to $K=4$ improved feature representation and, accordingly, the estimation results. However, at $K=5$, the performance declined, indicating that adding more heads beyond $K=4$ will not further improve the model (due to overfitting).

\begin{table}
\centering
\caption{CGNN-SE performance with varying heads within MH-GAT on IEEE 118-bus system}
\resizebox{0.85\columnwidth}{!}{
\renewcommand{\arraystretch}{1} 
\begin{tabular}{@{}ccc@{}}
\toprule 
\multirow{2}{*}{\textbf{\textbf{\begin{tabular}[c]{@{}c@{}}Num. of heads\\ within MH-GAT layer\end{tabular}}}}  & \multirow{2}{*}{\textbf{\begin{tabular}[c]{@{}c@{}}Magnitude\\ MAPE (\%)\end{tabular}}}  & \multirow{2}{*}{\textbf{\begin{tabular}[c]{@{}c@{}}Phase Angle\\ MAE (degrees)\end{tabular}}}  \\ \\ 
\midrule 
$K=1$ & 0.131 & 0.117 \\
$K=2$ & 0.056 & 0.044 \\
$K=3$ & 0.030 & 0.031 \\
$K=4$ & 0.018 & 0.027 \\
$K=5$ & 0.020 & 0.037 \\
\bottomrule
\end{tabular}
}
\label{table_ablation2}
\end{table}

\subsection{Performance under Different Sets of PMUs}
In our initial analysis, we chose a fixed PMU set of 11 for the 118-bus system and 120 for the 2000-bus system.
To further analyze the effectiveness of CGNN-SE, we have conducted additional experiments using a range of PMU sets.
The results of these evaluations are listed in Table \ref{table5} along with the results of our initial analysis (from Table \ref{table1}).
It can be observed from Table \ref{table5} that the proposed CGNN-SE performs consistently across different PMU sets.

\begin{table}[ht]
\centering
\captionsetup{justification=centering}
\caption{Performance of the proposed CGNN-SE with different PMU sets on IEEE 118-bus and 2000-bus Texas systems}
\resizebox{0.9\columnwidth}{!}{
\begin{tabular}{@{}cccc@{}}
\toprule
\multirow{2}{*}{\rotatebox[origin=c]{0}{\textbf{System}}} & \multirow{2}{*}{\textbf{\begin{tabular}[c]{@{}c@{}}Num. of buses\\ with PMUs\end{tabular}}} & \multirow{2}{*}{\textbf{\begin{tabular}[c]{@{}c@{}}Magnitude\\ MAPE (\%)\end{tabular}}} & \multirow{2}{*}{\textbf{\begin{tabular}[c]{@{}c@{}}Phase Angle\\ MAE (degrees)\end{tabular}}} \\ [10pt]
\midrule
\multirow{3}{*}{\rotatebox[origin=c]{0}{118-bus}}
& 11$^{1}$ & 0.018 & 0.027  \\
& 11$^{2}$ & 0.010 & 0.021  \\
& 32$^{3}$ & 0.014 & 0.021  \\
\midrule 
\multirow{3}{*}{\rotatebox[origin=c]{0}{2000-bus}}
& 120$^{4}$ & 0.093 & 0.064  \\
& 100$^{5}$ & 0.090 & 0.069  \\
& 120$^{6}$ & 0.101 &  0.069 \\
\midrule
\multicolumn{4}{@{}l@{}}{$^{1}$ {\footnotesize PMUs at 11 high-voltage (345 kV) buses (Initial analysis in 118-bus)}}\\
\multicolumn{4}{@{}l@{}}{$^{2}$ {\footnotesize PMUs at 11 random buses}}\\
\multicolumn{4}{@{}l@{}}{$^{3}$ {\footnotesize PMUs at 32 selected buses reported in \cite{tian2021neural}}}\\
\multicolumn{4}{@{}l@{}}{$^{4}$ {\footnotesize PMUs at 120 high-voltage (500 kV) buses (Initial analysis in 2000-bus)}}\\
\multicolumn{4}{@{}l@{}}{$^{5}$ {\footnotesize PMUs at 100 random buses}}\\
\multicolumn{4}{@{}l@{}}{$^{6}$ {\footnotesize PMUs at 120 random buses}}\\
\bottomrule
\end{tabular}%
}
\label{table5}
\end{table}

\vspace{-1.1em}

\subsection{Comparison with a State-of-the-art State Estimator}\label{WLM}
A complex-variable weighted least modulus (WLM) state estimator
for systems equipped with PMUs has recently been introduced in \cite{jabr2023complex}. Unlike traditional approaches that operate solely within the real domain, this estimator leveraged complex variables to enhance estimation performance.
Table \ref{table3} compares the performance of the proposed CGNN-SE with the WLM-SE on 1354pegase system considering the same PMU configuration reported in \cite{jabr2023complex} and in terms of the same accuracy metric that was used in \cite{jabr2023complex}, namely,
\begin{equation}\label{new_metric}
\sigma_y^2 = \sum_{i=1}^{N} \left|{\hat{y}}_i - y_i\right|^2    
\end{equation}
where, $N$ represents the number of nodes (buses) in the power system graph, and $y_i$ and $\hat{y_i}$ refer to the actual and estimated states in complex form, respectively.
From the results provided in Table \ref{table3}, it can be inferred that the proposed CGNN-SE has a smaller error in comparison to the WLM approach, suggesting that CGNN-SE is more accurate for state estimation on the 1354pegase system under same conditions.
These results also demonstrate the scalability of the proposed CGNN-SE and its effectiveness when applied to large-scale systems.

\vspace{-0.8em}

\begin{table}[hb]
\centering
\caption{Comparison of CGNN-SE with WLM-SE for 1354pegase system}
\tiny
\resizebox{0.63\columnwidth}{!}{%
\renewcommand{\arraystretch}{1.5} 
\begin{tabular}{@{}ccccc@{}}
\toprule 
  & WLM-SE \cite{jabr2023complex} & \textbf{CGNN-SE} \\ 
\midrule 
$\sigma_y^2$ & 0.011 &  \textbf{0.003}& \\
\bottomrule
\end{tabular}%
}
\label{table3}
\end{table}

\vspace{-1em}


\section{Conclusion and Future Work}
\label{Conclusions}

In this paper, we present CGNN-SE, an innovative deep geometric learning framework designed for estimating power system states when PMU observability is limited. By combining customized GCN and MH-GAT layers, CGNN-SE effectively captures complex node relationships and dependencies as well as ensures robustness against topology changes during real-time evaluation.
We also provide mathematical guarantees for the framework's adaptability to topology changes.
The CGNN-SE model intelligently integrates historical SCADA data to address challenges from PMU scarcity and real-time PMU failures, while also demonstrating robustness to non-Gaussian measurement noise—a major weakness of traditional optimization-based methods.
Moreover, the proposed CGNN-SE framework is found to work well in presence of bad data, attain PMU-timescale operation, and be scalable.
These features make CGNN-SE a versatile tool for achieving high-speed, accurate, and reliable situational awareness of power systems that are incompletely observed by PMUs.

Future work involves applying the proposed CGNN-SE approach to real-world power systems to assess its practical effectiveness under actual operating conditions.
Deriving an upper limit on the maximum errors produced by the CGNN-SE in absence of topology changes (similar to what was done for DNN-SE in \cite{azimian2024analytical}) can also be performed.
Additionally, the performance of CGNN-SE can be further evaluated in scenarios where the power system graph undergoes more complex evolutions
such as simultaneous topology changes and the presence of bad data.
Extending this approach to distribution systems is another promising direction, given their distinct structural and operational characteristics. Moreover, incorporating temporal dynamics with time-series data could enable the model to adapt to evolving grid conditions more effectively. Finally, we can explore broader applications of GNNs in power systems, leveraging their strong learning and generalization capabilities to enhance monitoring and decision-making processes.

\appendix
{\small
\setcounter{equation}{0}
\numberwithin{equation}{subsection}

\subsection{Proof of Theorem \ref{theorem2}}\label{appendix1}
\begin{proof}
The $i^{th}$ row and $j^{th}$ column of $AXW$ can be written as:
\begin{equation}\label{eq:app1}
(AXW)_{ij} = \sum_{f=1}^{F}\sum_{n=1}^{N}a_{in}x_{nf}w_{fj}
\end{equation}
where $a$, $x$, and $w$ are elements of matrices $A$, $X$, and $W$, respectively.
Since $X$ follows a mixture of Gaussians, \eqref{eq:app1} can also be expressed in the form
shown below:
\begin{align}\label{eq:app2}
\begin{aligned}
    &\sum_{f=1}^{F}\sum_{n=1}^{N}a_{in}x_{nf}w_{fj}\\ 
    &\sim \sum_{c=1}^{C}\pi_i^{c}\mathcal{N} \left (\sum_{f=1}^{F}\sum_{n=1}^{N}a_{in}m_{nf}^{c}w_{fj}, \sum_{f=1}^{F}\sum_{n=1}^{N}a_{in}^2 s_{nf}^{c}w_{fj}^2 \right)\\
    &= \sum_{c=1}^{C}\pi_i^{c}\mathcal{N} ({(AM^{c}W)}_{ij}, {\left ((A\odot A)S^{c}(W\odot W) \right)}_{ij} )
\end{aligned}
\end{align}

Now, setting $ \hat{s}_{ij}^{c}={((A\odot A)S^{c}(W\odot W) )}_{ij}$, and $\hat{m}_{ij}^{c}={(AM^{c}W)}_{ij}$ in \eqref{eq:app2}, and combining it with \eqref{eq:app1}, we get:
\begin{equation}\label{eq:app3}
    (AXW)_{ij} = \sum_{c=1}^{C}\pi_i^{c}\mathcal{N}(\hat{m}_{ij}^{c}, \hat{s}_{ij}^{c})
\end{equation}

Lastly, according to Theorem \ref{theorem1}, we get:
\begin{equation}\label{eq:app4}
\begin{aligned}
    &\ReLU[(AXW)_{ij}] = \sum_{c=1}^{C}\pi_i^{c}\ReLU \left [ \mathcal{N}(\hat{m}_{ij}^{c}, \hat{s}_{ij}^{c}) \right]\\ 
    &= \sum_{c=1}^{C}\pi_i^{c}\sqrt{\hat{s}_{ij}^{c}}\NR{\left( \frac{\hat{m}_{ij}^{c}}{\sqrt{\hat{s}_{ij}^{c}}}  \right)}
\end{aligned}
\end{equation}

\end{proof}

\setcounter{equation}{0}
\numberwithin{equation}{subsection}

\vspace{-2.3em}

\subsection{Proof of Theorem \ref{theorem3}}\label{appendix2}

\begin{proof}
Let $X^{(l)}=\sigma\left(AX^{(l-1)}{W^\top}^{(l)}\right)$ be the output of the mapping function $\Phi(A,X,W)$ at layer $l$.
We consider columns of $X^{(l)} \in \R^{N \times f^{(l)}}$ as graph signals $\{{x^{(l)}_{:c}}\}_{c=1}^{f^{(l)}} \in \R^{N}$ such that at every $l$ there exist $f^{(l)}$ graph signals, each defined based on the graph signals from the previous layer as shown below:
\begin{equation}\label{eq:app5}
x^{(l)}_{:c} =\sigma \left( \sum_{g=1}^{f^{(l-1)}} w^{(l)}_{cg} A  {x^{(l-1)}_{:g}} \right) 
\end{equation}

Here, $w^{(l)}_{cg}$ denotes the elements of the learnable weight matrix $W \in \R^{f^{(l)} \times f^{(l-1)}}$, and $x^{(l-1)}_{:g}$ is the graph signals at layer $l-1$ such that $c=1, \dots, f^{(l)}$ and $g=1, \dots, f^{(l-1)}$.
Now, we assess the upper bound of the error between the final outputs of the mapping function $\Phi(\cdot)$ at layer $L$, after a topology change $\left(\{A,X\} \rightarrow \{A',X'\}\right)$.
This assessment takes into account that the output of $\Phi(\cdot)$ represents a collection of graph signals $\{{x^{(L)}_{:c}}\}_{c=1}^{f^{(L)}}$ as shown below \cite{gama2020stability}:
\begin{equation}\label{eq:app6}
{\left\lVert \Phi(A,X,W) - \Phi(A',X',W) \right\rVert}^2 = \sum_{c=1}^{f^{(L)}} {\left\lVert x^{(L)}_{:c} - x'^{(L)}_{:c} \right\rVert}^2
\end{equation}
where, $x'^{(L)}_{:c}$ denotes the graph signals at the final layer $L$ after a topology change.
Considering \eqref{eq:app5}, and replacing $w^{(L)}_{cg}A$ and $w^{(L)}_{cg}A'$ by $H^{(L)}_{cg}$ and $H'^{(L)}_{cg}$, respectively, 
we obtain the following for the right-hand side norm of \eqref{eq:app6}:
\begin{equation}\label{eq:app7}
\begin{aligned}
    &{\left\lVert x^{(L)}_{:c} - x'^{(L)}_{:c} \right\rVert} \\
    &=  \left \lVert \sigma \left (\sum_{g=1}^{f^{(L-1)}}H^{(L)}_{cg} x^{(L-1)}_{:g} \right ) - \sigma \left(\sum_{g=1}^{f^{(L-1)}} H'^{(L)}_{cg} x'^{(L-1)}_{:g}\right )  \right \lVert
\end{aligned}  
\end{equation}

Considering the Lipschitz\footnote{The non-linearity $\sigma(\cdot)$, which satisfies $\sigma(0) = 0$, is Lipschitz if, for any $a$, $b \in \R$, there exists a
constant $C_{\sigma} > 0$ such that $\left|\sigma(a)-\sigma(b)\right| \leq C_{\sigma} \left|a-b\right|$.}
continuity of non-linearity $\sigma(\cdot)$,
and applying the triangle inequality, we obtain:

\begin{equation}\label{eq:app8}
\begin{aligned}
    & {\left\lVert x^{(L)}_{:c} - x'^{(L)}_{:c} \right\rVert} \\
    & \leq C_{\sigma} \sum_{g=1}^{f^{(L-1)}} \left\lVert H^{(L)}_{cg} x^{(L-1)}_{:g} -  H'^{(L)}_{cg} x'^{(L-1)}_{:g} \right\lVert
\end{aligned} 
\end{equation}

Adding and subtracting $H'^{(L)}_{cg} x^{(L-1)}_{:g}$ within the sum, and then reapplying the triangle inequality, for the right-hand side norm of \eqref{eq:app8}, we get:
\begin{equation}\label{eq:app9}
\begin{aligned}
    & \left\lVert H^{(L)}_{cg} x^{(L-1)}_{:g} -  H'^{(L)}_{cg} x'^{(L-1)}_{:g} \right\lVert \\
    & \leq \left\lVert \left (H^{(L)}_{cg} -  H'^{(L)}_{cg} \right) x^{(L-1)}_{:g}\right\lVert + \left\lVert H'^{(L)}_{cg} \left(x^{(L-1)}_{:g} -  x'^{(L-1)}_{:g} \right)\right\lVert
\end{aligned}  
\end{equation}

Now, using the sub-multiplicative property of the operator norm,
we can rewrite \eqref{eq:app9} as:
\begin{equation}\label{eq:app10}
\begin{aligned}
    & \left\lVert H^{(L)}_{cg} x^{(L-1)}_{:g} -  H'^{(L)}_{cg} x'^{(L-1)}_{:g} \right\lVert \\
    & \leq \left\lVert H^{(L)}_{cg} -  H'^{(L)}_{cg}\right\lVert \left\lVert x^{(L-1)}_{:g}\right\lVert  +    \left\lVert H'^{(L)}_{cg}\right\lVert \left\lVert x^{(L-1)}_{:g} -  x'^{(L-1)}_{:g} \right\lVert
\end{aligned}    
\end{equation}

As $\left\| A - A' \right\| \leq \epsilon$ and $\left| {w^{(l)}_{ij}} \right| \leq \delta$ for all layers $l=1, \dots, L$,
we can replace the first term in right-hand side of \eqref{eq:app10} by $\delta\epsilon  \left\lVert x^{(L-1)}_{:g}\right\lVert$.
For the second term, we make use of the fact that $\left\lVert H'^{(L)}_{cg}\right\lVert \leq B$ 
to obtain the following relation from \eqref{eq:app8}:
\begin{equation}\label{eq:app11}
\begin{aligned}
    & \left\lVert x^{(L)}_{:c} - x'^{(L)}_{:c} \right\rVert \\
    & \leq C_{\sigma} \sum_{g=1}^{f^{(L-1)}} \left(\delta\epsilon \left\lVert x^{(L-1)}_{:g} \right\lVert + B \left\lVert x^{(L-1)}_{:g} -  x'^{(L-1)}_{:g} \right\lVert \right)
\end{aligned}    
\end{equation}

We observe that \eqref{eq:app11} illustrates a recursive relationship, where the bound at layer $L$ depends on the bound at layer $L - 1$ and the norm of the features at layer $L - 1$ summed up over all features. Now, for any arbitrary layer $l \in \{1, \dots, L\}$, we can rewrite this recursion as:
\begin{equation}\label{eq:app12}
\begin{aligned}
    & \left\lVert x^{(l)}_{:c} - x'^{(l)}_{:c} \right\rVert \\
    & \leq C_{\sigma} \sum_{g=1}^{f^{(l-1)}} \left(\delta\epsilon \left\lVert x^{(l-1)}_{:g} \right\lVert + B \left\lVert x^{(l-1)}_{:g} -  x'^{(l-1)}_{:g} \right\lVert \right)
\end{aligned}    
\end{equation}

To simplify the right-hand side of \eqref{eq:app12}, we express
the norm $\left\lVert x^{(l)}_{:c} \right\rVert$ as:
\begin{equation}\label{eq:app13}
\begin{aligned}
    &\left\lVert x^{(l)}_{:c} \right\rVert = \left\lVert \sigma \left( \sum_{g=1}^{f^{(l-1)}} H^{(l)}_{cg} {x^{(l-1)}_{:g}} \right)\right\rVert  \leq C_{\sigma} \left\lVert \sum_{g=1}^{f^{(l-1)}} H^{(l)}_{cg} x^{(l-1)}_{:g} \right\lVert \\
    & \leq  C_{\sigma}B  \sum_{g=1}^{f^{(l-1)}}  \left\lVert x^{(l-1)}_{:g} \right\lVert \leq \left (C_{\sigma}B \right)^{l} \prod_{l'=1}^{l-1}f^{(l')}\sum_{g=1}^{f^{(0)}}\left\lVert x^{(0)}_{:g} \right\lVert
\end{aligned}    
\end{equation}

In \eqref{eq:app13}, we make use of \eqref{eq:app5}, followed by the triangle inequality, and the recursion relationship considering $\left\lVert x^{(0)}_{:g}\right\lVert$ as the initial graph signals for $g=1,\dots,f^{(0)}$.
Lastly, as topology changes have localized effects on the graph signals \cite{ejebe1992adaptive}, the term $\left\lVert x^{(l-1)}_{:g} -  x'^{(l-1)}_{:g} \right\lVert$ in \eqref{eq:app12} becomes negligibly small.
Implementing these changes in \eqref{eq:app12}, we obtain:
\begin{equation}\label{eq:app14}
\begin{aligned} 
    & \left\lVert x^{(l)}_{:c} - x'^{(l)}_{:c} \right\rVert \\
    & \leq \delta\epsilon \left (C_{\sigma}B \right)^{l-1}  \left(\sum_{l'=1}^{l}C^{l'}_{\sigma} \right) \left (\prod_{l'=1}^{l-2}f^{(l')}\right) \sum_{g=1}^{f^{(0)}}\left\lVert x^{(0)}_{:g} \right\lVert
\end{aligned}    
\end{equation}

Applying $l=L$ and using it back in \eqref{eq:app6}, we obtain:
\begin{equation}\label{eq:app15}
\begin{aligned}
    &{\lVert \Phi(A,X,W) - \Phi(A',X',W) \rVert}^2 = \\
    &\leq \sum_{c=1}^{f^{(L)}} {\left(\delta\epsilon \left (C_{\sigma}B \right)^{L-1}   \sum_{l'=1}^{L}C^{l'}_{\sigma} \prod_{l'=1}^{L-2}f^{(l')} \sum_{g=1}^{f^{(0)}}\left\lVert x^{(0)}_{:g} \right\lVert \right)}^2
\end{aligned}    
\end{equation}

Since no term in \eqref{eq:app15} depends on $c$, we simplify the summation over $c$ and then apply the square root to obtain:
\begin{equation}\label{eq:app16}
\begin{aligned}
    &\lVert \Phi(A,X,W) - \Phi(A',X',W) \rVert \\
    &\leq \sqrt{f^{(L)}}\delta\epsilon \left (C_{\sigma}B \right)^{L-1}   \sum_{l'=1}^{L}C^{l'}_{\sigma} \prod_{l'=1}^{L-2}f^{(l')} \sum_{g=1}^{f^{(0)}}\left\lVert x^{(0)}_{:g} \right\lVert
\end{aligned}    
\end{equation}

Finally, setting $f^{(0)} = f^{(L)}=2$ (as there are two features to be estimated),
$C_{\sigma} = 1$ (Lipshitz constraint), $f^{(1)}=\dots=f^{(L-2)}=F$, and $\lambda = \sum_{g=1}^{2}\left\lVert x^{(0)}_{:g} \right\lVert$, the following (final) bound is obtained, which completes the proof:
\begin{equation}\label{eq:app17}
\lVert \Phi(A,X,W) - \Phi(A',X',W) \rVert \leq  \sqrt{2}\lambda\delta\epsilon L B^{L-1} F^{L-2}
\end{equation}
\end{proof}
}

\vspace{-2.2em}

\bibliographystyle{ieeetr}
{\footnotesize
\bibliography{References.bib}}




\begin{IEEEbiography}[{\includegraphics[width=1in,height=1.25in,clip,keepaspectratio]{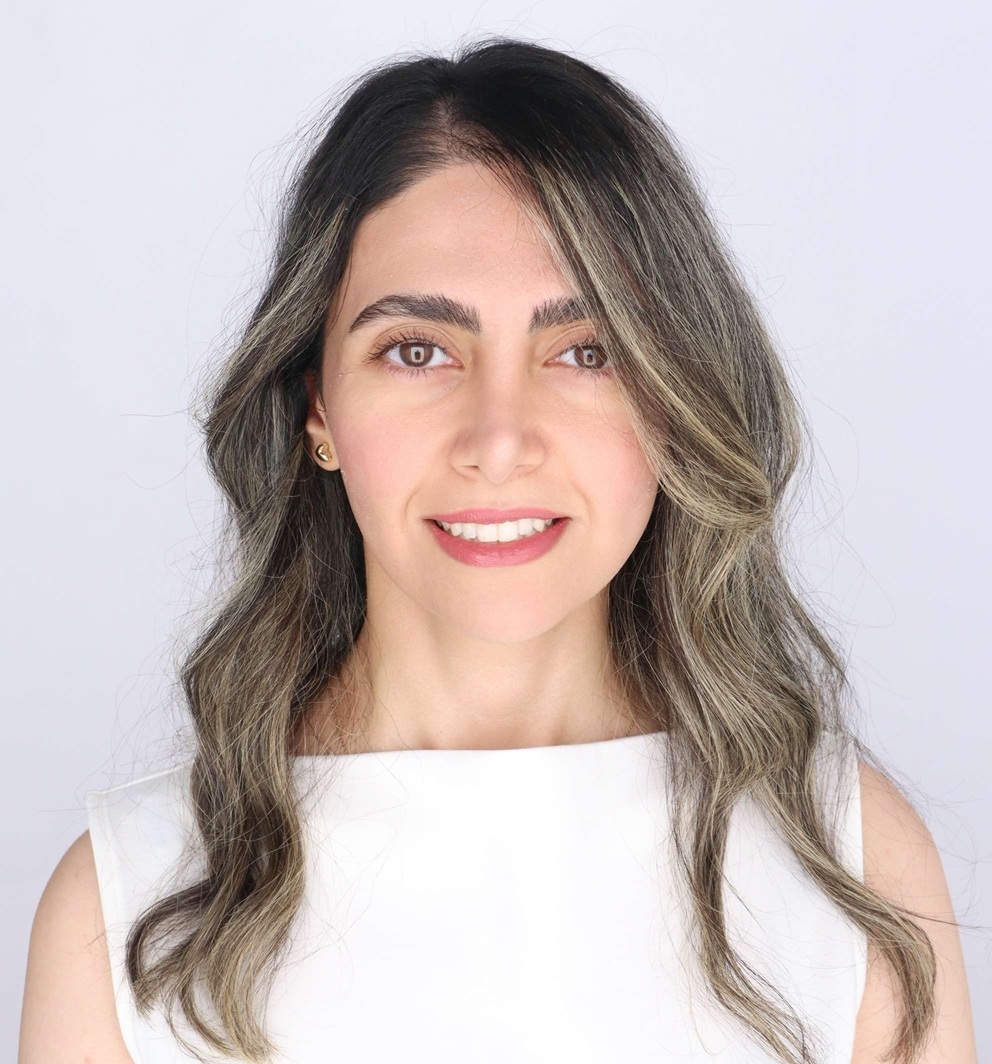}}] {Shiva Moshtagh} (Student Member, IEEE) received the B.Sc. degree in electrical engineering from the Jundi-Shapur University of Technology, Dezful, Iran, in 2015, and the M.Sc. degree in electrical engineering from Imam Khomeini International University, Qazvin, Iran, in 2019. She is currently pursuing the Ph.D. degree in electrical engineering with Arizona State University, Tempe, AZ, USA. Her research interests include power system monitoring, estimation, and optimization. In recent years, she has focused on applying machine learning-based
methods to power systems.
\end{IEEEbiography}

\begin{IEEEbiography}[{\includegraphics[width=1in,height=1.25in,clip,keepaspectratio]{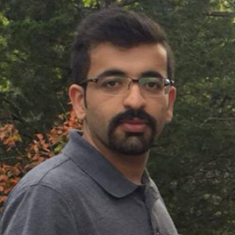}}] {Behrouz Azimian} received the B.Sc. degree in electrical engineering from the Iran University of Science and Technology, Tehran, Iran, in 2016, the M.Sc. degree in electrical engineering from Alfred University, NY, USA, in 2019, and the Ph.D. degree in electrical engineering from Arizona State University, AZ, USA, in 2024. He is currently a Senior Electricity Market Software Engineer at GE Vernova. His research interests include machine learning, deep learning, and the application of artificial intelligence and optimization techniques in power systems, including state estimation and electricity markets.
\end{IEEEbiography}

\begin{IEEEbiography}[{\includegraphics[width=1in,height=1.25in,clip,keepaspectratio]{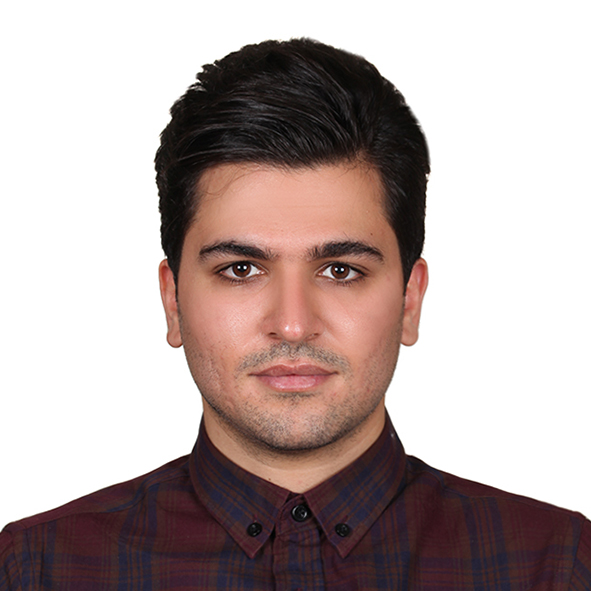}}] {Mohammad Golgol} (Student Member, IEEE) received the B.Sc. degree in electrical engineering from Jundi-Shapur University of Technology, Dezful, Iran, in 2015, and the M.Sc. degree in electrical engineering from K. N. Toosi University of Technology, Tehran, Iran, in 2019. He is currently pursuing the Ph.D. degree in electrical engineering at Arizona State University, Tempe, AZ, USA. His research interests include machine learning applications, optimization and control techniques in power systems. Recently, he has primarily focused on leveraging machine learning techniques including, deep reinforcement learning, and generative models to enhance power system stability and efficiency, facilitate renewable energy integration, and optimize electric vehicle charging coordination.
\end{IEEEbiography}

\begin{IEEEbiography}[{\includegraphics[width=1in,height=1.25in,clip,keepaspectratio]{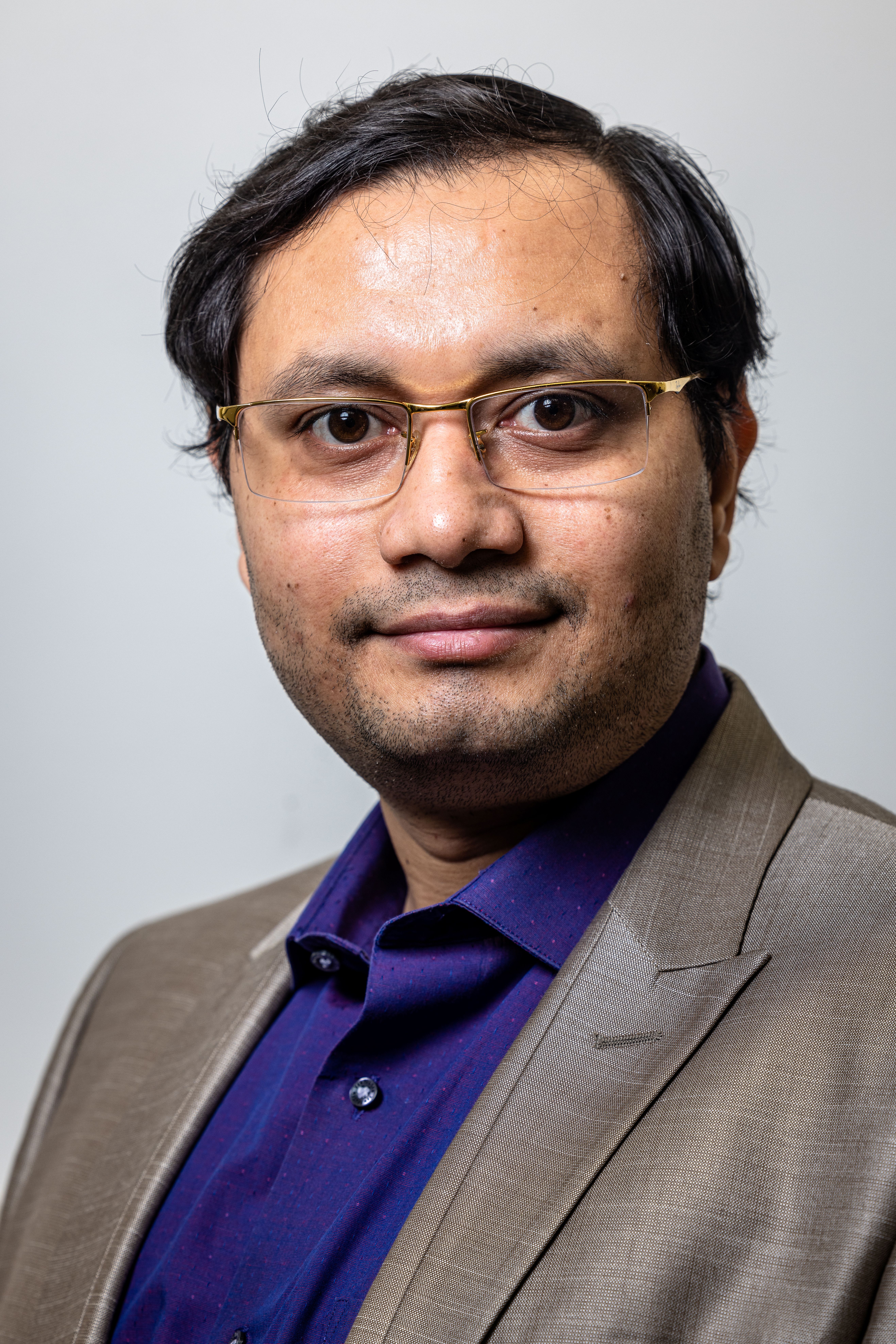}}] {Anamitra Pal} (Senior Member, IEEE) received the B.E. degree (summa cum laude) in electrical and electronics engineering from the Birla Institute of Technology at Mesra, Ranchi, India, in 2008, and the M.S. and Ph.D. degrees in electrical engineering from Virginia Tech, Blacksburg, VA, USA, in 2012 and 2014, respectively. He is currently an Associate Professor in the School of Electrical, Computer, and Energy Engineering at Arizona State University (ASU), Tempe, AZ, USA. His research interests include data analytics with a special emphasis on time-synchronized measurements, artificial intelligence (AI)-applications in power systems, and critical infrastructure resiliency assessment. Dr. Pal has received the 2018 Young CRITIS Award for his valuable contributions to the field of critical infrastructure protection, the 2019 Outstanding Young Professional Award from the IEEE Phoenix Section, the 2022 NSF CAREER Award, and the Centennial Professorship Award from ASU in 2023.
\end{IEEEbiography}




\vfill

\end{document}